\documentclass[letterpaper,11pt]{article}
\usepackage[margin=1in]{geometry}
\usepackage[bookmarks, colorlinks=true, plainpages = false, colorlinks=true,
            citecolor=red,
            linkcolor=blue,
            anchorcolor=red,
            urlcolor=blue]{hyperref}
\usepackage{url}
\usepackage{amsmath,amsfonts,amsthm,amssymb,bm, verbatim,dsfont,mathtools}
\usepackage{color,graphicx,appendix}
\usepackage{etoolbox}

\usepackage{array}
\usepackage{multirow}

\usepackage{float}
\usepackage{tikz}
\definecolor{rnodecolor}{RGB}{0,114,178}
\definecolor{unodecolor}{RGB}{213,94,0}
\tikzset{queuevertex/.style={circle,draw=black,line width=0.4pt,minimum size=6pt,inner sep=0pt,outer sep=0.4pt}}
\usetikzlibrary{matrix,arrows,calc,shapes,backgrounds}
\usetikzlibrary{shapes.callouts,decorations.text}

\usetikzlibrary{shapes.misc}

\tikzset{cross/.style={cross out, draw=black, minimum size=2*(#1-\pgflinewidth), inner sep=0pt, outer sep=0pt},
cross/.default={1pt}}

\tikzstyle{int}=[draw, fill=blue!20, minimum size=2em]
\tikzstyle{dot}=[circle, draw, fill=blue!20, minimum size=2em]
\tikzstyle{dotred}=[circle, draw, fill=red!20, minimum size=2em]
\tikzstyle{init} = [pin edge={to-,thin,black}]
\tikzstyle{initred} = [pin edge={to-,thin,red}]
\tikzstyle{plan}=[draw, fill=blue!20, minimum size=2em, text width=5em, rounded corners,align=center]
\tikzstyle{planwide}=[draw, fill=blue!20, minimum size=2em, text width=8em, rounded corners,align=center]
\tikzstyle{nodedot}=[circle, draw, fill=white, minimum size=0.3cm,inner sep=0pt]
\tikzstyle{nodedot}=[circle, draw, fill=white, minimum size=3,inner sep=0pt]
\tikzstyle{Medge}=[green!60!black, thick]
\tikzstyle{Bedge}=[red, thick]
\tikzstyle{Cedge}=[blue, thick]
\tikzstyle{Sedge}=[black, thick]
\tikzstyle{Mgiantedge}=[green!60!black, line width=3.0pt]
\tikzstyle{Bgiantedge}=[red,line width=3.0pt]
\tikzstyle{Cgiantedge}=[blue,line width=3.0pt]
\tikzstyle{Sgiantedge}=[black,line width=3.0pt]
\tikzstyle{shadedgiantnode}=[circle, draw, fill=black!10, minimum size=1cm, inner sep=0pt]
\tikzstyle{unshadedgiantnode}=[circle, draw, fill=white, minimum size=1cm, inner sep=0pt]
\makeatletter
\tikzset{my loop/.style =  {to path={
  \pgfextra{}
  [looseness=5,min distance=10mm]
  \tikz@to@curve@path},font=\sffamily\small
  }}
\makeatletter
\newcolumntype{C}[1]{>{\centering\arraybackslash}p{#1}}

\tikzstyle{vertexdot}=[circle, draw, fill=white, minimum size=3,inner sep=0pt]
\tikzstyle{root}=[circle, draw, fill=black, minimum size=3,inner sep=0pt]

\tikzstyle{vertexdotsolid}=[circle, draw, fill=black, minimum size=3,inner sep=0pt]

\usepackage{pgfplots}

\pgfplotsset{
    standard/.style={
        axis x line=middle,
        axis y line=middle,
        every axis x label/.style={at={(current axis.right of origin)},anchor=north west},
        every axis y label/.style={at={(current axis.above origin)},anchor=north west}
    }
}
\usepgfplotslibrary{fillbetween}

\usepackage{xr,xspace}
\usepackage{todonotes}
\usepackage{paralist}
\usepackage{enumitem}

\usepackage{caption,subcaption,soul}
\usepackage{algorithm}
\usepackage{algpseudocode}
\makeatletter
\usepackage{romanbar}
\usepackage{booktabs}
\usepackage{cleveref}
\usepackage{array}

\theoremstyle{plain}
\newtheorem{theorem}{Theorem}
\newtheorem{lemma}{Lemma}
\newtheorem{proposition}{Proposition}
\newtheorem{corollary}{Corollary}
\newtheorem{claim}{Claim}

\theoremstyle{definition}
\newtheorem{definition}{Definition}

\newtheorem{problem}{Problem}

\newtheorem*{remark*}{Remark}
\newtheorem*{theorem*}{Theorem}

\newcommand{\1}[1]{{\mathbf{1}_{\left\{{#1}\right\}}}}

\usepackage{xspace,prettyref}

\newcommand{\zeros}{\mathbf 0}
\newcommand{\reals}{{\mathbb{R}}}

\newcommand{\Expect}{\mathbb{E}}

\def\Var{\mathrm{Var}}

\newrefformat{eq}{(\ref{#1})}
\newrefformat{chap}{Chapter~\ref{#1}}
\newrefformat{sec}{Section~\ref{#1}}
\newrefformat{alg}{Algorithm~\ref{#1}}
\newrefformat{fig}{Figure~\ref{#1}}
\newrefformat{tab}{Table~\ref{#1}}
\newrefformat{rmk}{Remark~\ref{#1}}
\newrefformat{clm}{Claim~\ref{#1}}
\newrefformat{claim}{Claim~\ref{#1}}
\newrefformat{def}{Definition~\ref{#1}}
\newrefformat{cor}{Corollary~\ref{#1}}
\newrefformat{lmm}{Lemma~\ref{#1}}
\newrefformat{prop}{Proposition~\ref{#1}}
\newrefformat{prob}{Problem~\ref{#1}}
\newrefformat{app}{Appendix~\ref{#1}}
\newrefformat{hyp}{Hypothesis~\ref{#1}}
\newrefformat{thm}{Theorem~\ref{#1}}
\newrefformat{fac}{Fact~\ref{#1}}
\newrefformat{ex}{Example~\ref{#1}}

\newcommand{\norm}[1]{\left\|{#1} \right\|}

\newcommand{\calA}{{\mathcal{A}}}

\newcommand{\calC}{{\mathcal{C}}}
\newcommand{\calD}{{\mathcal{D}}}

\newcommand{\calG}{{\mathcal{G}}}

\newcommand{\calM}{{\mathcal{M}}}

\newcommand{\calR}{{\mathcal{R}}}

\newcommand{\calT}{{\mathcal{T}}}

\renewcommand{\bar}{\overline}

\newcommand{\regret}{\mathrm{Regret}}

\newcommand{\TP}{\mathbf{TP}}
\newcommand{\CP}{\mathbf{CP}}
\newcommand{\GP}{\mathbf{GP}}

\newcommand{\SPP}{\mathrm{SPP}}

\newif\ifshowrevisions
\showrevisionstrue

\definecolor{revgreen}{RGB}{0,110,40}

\title{Optimal regret for dynamic matching with static priorities}

\author{S\"uleyman Kerimov, Pengyu Qian, Mingwei Yang, and Sophie H.\ Yu\thanks{S. Kerimov is with the Jones Graduate School of Business, Rice University, Houston TX, USA,
\texttt{kerimov@rice.edu}.
P. Qian is with the Questrom School of Business, Boston University, Boston MA, USA, \texttt{pqian@bu.edu}.
M. Yang is with the Department of Management Science and Engineering, Stanford University, Stanford CA, USA,
\texttt{mwyang@stanford.edu}.
S.\ H.\ Yu is with The Wharton School, University of Pennsylvania, Philadelphia PA, USA,  \texttt{hysophie@wharton.upenn.edu}.
}}

\date{\today}

\begin{document}
\maketitle

\begin{abstract}
We study centralized two-way dynamic matching markets in which agents arrive stochastically over time and wait in queues associated with their types. We focus on static priority policies that assign each agent type a fixed ranking over its neighbors and make online matching decisions using only queue availability. We propose the General Priority ($\GP$) policy, which builds on the Tree Priority ($\TP$) policy introduced in \cite{kerimov2025optimality}. Under the general position gap condition, $\GP$ selects matches using a nondegenerate optimal basic solution of the fluid relaxation. It prioritizes children over parents within trees and consistently gives clockwise neighbors priority over counterclockwise neighbors on each cycle. We establish the optimal all-time regret scaling  $O(\epsilon^{-1})$ for $\GP$ on general two-way matching networks, where the general position gap $\epsilon$ measures the fluid relaxation's distance from degeneracy. In particular, this resolves the conjecture by \cite{kerimov2025optimality} on the optimal $\epsilon$-dependence for $\TP$.
\end{abstract}

\tableofcontents

\section{Introduction}

Dynamic matching markets, where agents arrive over time and must be matched based on compatibility, arise in numerous applications including kidney exchange markets (e.g., \cite{ashlagi2018effect}), online carpooling and ride-hailing platforms (e.g., \cite{ozkan2020dynamic}), and logistics and delivery systems (e.g., \cite{eom2025batching}). In such settings, agents of different types arrive stochastically to type-specific queues, waiting to be matched. The allowable matches are captured by an underlying compatibility graph, where edge weights encode edge-specific matching rewards. A central design problem is to construct online greedy policies 
that perform well uniformly over time while balancing immediate and future rewards, while remaining simple enough to implement in large, distributed systems.

A substantial body of recent literature analyzes dynamic matching policies relative to an offline omniscient benchmark, with a particular focus on the \emph{general position gap} (GPG) condition. This condition is generically satisfied when the optimal basic matches (those corresponding to a basic optimal solution of the fluid relaxation)
remain stable under small perturbations of the model parameters \cite{megiddo1989varepsilon, jasin2012, DBLP:journals/ior/ChenL024}. The GPG parameter $\epsilon$ quantifies the magnitude of such allowable perturbations, effectively measuring the distance of the fluid relaxation from degeneracy~\cite{wei2026constant,kerimov2024dynamic,gupta2024greedy,kerimov2025optimality,he2025online}. Intuitively, $\epsilon$ captures both the inherent ``thickness'' of the market and its operational stability: a larger $\epsilon$ implies a wider separation between optimal and suboptimal actions, making it easier for the central planner to consistently perform ``correct'' matches. Understanding how a policy's all-time regret scales with $\epsilon$ is therefore crucial, as this relationship reveals the robustness of the policy across diverse operating environments.

Under the GPG condition, the optimal all-time regret, defined as the maximum expected difference between the performances of the offline optimum and the policy at any moment, is known to scale as $\Theta(\epsilon^{-1})$~\cite{kerimov2024dynamic}. Achieving this \textit{optimal regret scaling} has been the focus of substantial recent work. Policies such as longest-queue \cite{kerimov2025optimality}, sum-of-squares \cite{gupta2024greedy}, and primal-dual \cite{wei2026constant} have been extensively studied and shown to achieve optimal scaling. These policies typically require real-time access to exact queue-length
information. We ask whether fixed local priority rules, whose matching
decisions use only neighboring queue availability, can achieve the same
$O(\epsilon^{-1})$ regret scaling on general two-way matching networks.

A natural class of policies with reduced information requirements is that of static priority policies. Such policies perform matches based on a predetermined priority order over the matches. When the compatibility network is acyclic, such policies can be represented as a single global ordering over matches, and when the network is cyclic, each agent type has a local ordering over its neighbors. Importantly, once these priorities are fixed, the policy makes decisions based solely on whether queues are empty or not, rather than how long they are. Indeed, such availability information is the minimal information requirement for executing a match.

The question of whether one can design matching policies using less information without sacrificing performance is not merely of theoretical interest. It is driven by the operational costs (storing and communicating information) inherent in maintaining fine-grained state information. In large-scale load balancing, policies requiring instantaneous queue lengths incur prohibitive communication overheads compared to those using partial information \cite{mitzenmacher2001power,lu2011joinidlequeue,der2022scalable}.  In the context of dynamic matching, static priority policies have gained increasing attention and are particularly appealing for systems, since they are easy to describe and simple to implement~\cite{moyal2017instability,aveklouris2025matching,kerimov2025optimality}.

For two-way matching networks whose active graph is acyclic, \cite{kerimov2025optimality} proposed the tree priority ($\TP$) policy and established bounded all-time regret. Here, the active graph consists of the matches used by a chosen optimal basic solution of the static-planning problem. Under the GPG condition, each tree component has a unique under-demanded type, which serves as its root. $\TP$ gives children priority over parents. Although its all-time regret was known to be bounded, the dependence of this bound on the general position gap $\epsilon$ remained unresolved.

Extending this guarantee beyond trees raises a further challenge. Each connected component of the active graph is either a tree or an odd cycle with trees attached to its vertices. A component containing a cycle has no under-demanded type, so the root used by $\TP$ is unavailable. It is therefore necessary to determine both whether $\TP$ attains the optimal regret scaling on trees and how to extend its priority construction to cyclic active components. These issues lead to our central question: can static priority policies achieve the optimal all-time regret scaling for general two-way matching networks under the GPG condition? We use static priorities in the local sense: each agent type has a fixed ranking over its active neighbors.

We establish this optimal scaling through two results. First, we prove that the existing $\TP$ policy achieves $O(\epsilon^{-1})$ all-time regret on trees, sharpening its previously known guarantee of bounded regret by resolving its dependence on $\epsilon$. Second, we introduce the General Priority ($\GP$) policy, which extends $\TP$ to active components containing cycles, and prove that it achieves $O(\epsilon^{-1})$ all-time regret for general two-way matching networks under the GPG condition. $\GP$ coincides with $\TP$ on tree components.

The extension retains $\TP$’s priorities within the trees attached to a cycle and adds a consistently oriented priority rule around the cycle. Each attached tree is rooted at its cycle vertex. Arrivals at vertices outside the cycle give children priority over their parent. At each cycle vertex, arrivals prioritize attached children, followed by the clockwise neighbor and then the counterclockwise neighbor. Once the active graph and these priority lists are fixed, matching decisions require only the nonemptiness of neighboring queues. Thus, fixed local priorities suffice to achieve the optimal dependence of all-time regret on $\epsilon$.

Our analysis uses the tree and cycle structure in different ways. On trees, we bound alternating sums of queue lengths using reflected random walks with negative drift. For components containing a cycle, we combine these tree estimates with an auxiliary process that successively drains cycle queues. Nonexpansive coupling and a quadratic drift argument then yield the required uniform queue bounds and the optimal regret scaling.

\subsection{Related literature}

Dynamic matching has been extensively studied in various settings. We review several streams of literature most relevant to our work.

\paragraph{Multi-way dynamic matching.}
Starting from \cite{kerimov2024dynamic}, the multi-way dynamic matching problem has recently received extensive attention, and the optimal regret scaling of $O(\epsilon^{-1})$ has been achieved via several policies.
\cite{kerimov2024dynamic} propose a batching policy, which performs a maximum weighted matching periodically. 
\cite{gupta2024greedy} develops a sum-of-squares policy that achieves an optimal regret scaling even under a more general model.
Later on, \cite{wei2026constant} design a policy based on the primal-dual framework, which achieves constant regret at all times under unknown arrival rates and optimal regret scaling when the arrival rates are known.
\cite{kerimov2025optimality} consider two-way matching networks and focus on policies that perform a match whenever possible.
They show that the longest-queue policy achieves the optimal regret scaling, and they propose the $\TP$ policy that achieves constant regret at all times whenever the active graph is acyclic, without establishing the optimal dependence of the regret bound on $\epsilon$.

\paragraph{Dynamic matching on fixed graphs.}
Similar to our work, numerous papers study the dynamic matching problem with a fixed matching configuration by imposing different modeling assumptions.
Much of this literature studies matches consisting of exactly two agents, where the underlying matching network can be either bipartite~\cite{ozkan2020dynamic,kessel2022stationary,aveklouris2025matching,DBLP:conf/stoc/AmaniHamedaniAS25}
or non-bipartite~\cite{collina2020dynamic,aouad2022dynamic,DBLP:journals/corr/abs-2507-04551,arnosti2026dynamic}.
Performance criteria in models with stochastic abandonment include match-specific rewards~\cite{collina2020dynamic}, matching rewards net of holding costs~\cite{aveklouris2025matching}, and matching losses due to abandonment~\cite{DBLP:journals/mansci/KohlenbergG25}.
Related work studies holding-cost minimization in matching queues without abandonment~\cite{buyic2015approximate}.

Notably, \cite{DBLP:journals/corr/abs-2507-04551,arnosti2026dynamic} also consider static priority policies for dynamic matching.
Like $\GP$, their policies fix a ranked list of predetermined neighbors for each arriving type and match a new agent to the highest-ranked available neighbor.
However, they assume agent departure and aim to provide multiplicative approximation guarantees.
In addition, their policies determine priorities through value-based constructions, while $\GP$ constructs its priority lists using the tree and cycle structure of the active graph.

\paragraph{Dynamic matching on random graphs.}
There is a vast literature on dynamic matching models based on random graphs, where agents arrive over time and can form edges with the existing agents with fixed probabilities~\cite{anderson2013efficient,burqOR,akbarpour2020thickness}.
A common assumption in this literature is that match values are homogeneous, so the objective is to minimize the number of unmatched agents.
The performance of greedy matching depends on the information available to the planner.
In particular, \cite{akbarpour2020thickness} show that greedy matching is close to optimal when imminent departures are unobservable, whereas waiting to thicken the market can substantially reduce unmatched departures when such information is available.
Recently, \cite{blanchet2022asymptotically} consider a setting where match values are heterogeneous, and as the market grows large, they show that greedy threshold policies are asymptotically optimal.
Another line of research assumes that the agents lie in a metric space, and the weight of an edge between two agents is determined by the distance between them~\cite{kanoria2021dynamic,balkanski2023power,DBLP:journals/ior/YangY26,DBLP:conf/innovations/LiV026,yang2026optimal}.

\paragraph{Network revenue management.}
In the network revenue management (NRM) problem, there exist offline resources, and online requests that consume certain amounts of offline resources arrive dynamically.
\cite{talluri1998analysis} achieve $O(\sqrt{T})$ regret in the (quantity-based) NRM model via the bid-price policy.
Later on, \cite{jasin2012} reduce regret to be $O(1)$ via a re-solving policy under the GPG condition.
Subsequent work has established constant regret when the arrival rates are unknown~\cite{jasin2015performance}, when the GPG condition does not hold~\cite{arlotto2019uniformly,vera2021bayesian,bumpensanti2020re,vera2021online,li2026firstorder,vera2026erratum}, and for a variety of other related problems~\cite{balseiro2024survey}.
\cite{he2025online} develop a primal-dual policy that achieves constant regret under the GPG condition while solving the fluid LP only once at initialization.
For NRM with continuous reward distributions whose densities are bounded away from zero, \cite{jiang2022degeneracy} establish $O(\log^2 T)$ regret without a nondegeneracy assumption. More recently, \cite{zhang2026continuous} study how distributional regularity governs regret under degeneracy when resource consumption is also continuously distributed.
In contrast, for the dynamic matching problem considered in this paper, without the GPG condition, a regret lower bound of $\Omega(\sqrt{T})$ exists~\cite{kerimov2024dynamic}.

\paragraph{Stability of stochastic matching systems.}
The multi-way dynamic matching model is closely related to the study of stability of stochastic matching systems~\cite{mairesse2016stability,rahme2021stochastic,comte2021online,jonckheere2023generalized}.
The connection is established by, e.g., \cite[Lemma 4.1]{kerimov2024dynamic} and \cite[Lemma 1]{gupta2024greedy}, which assert that bounded all-time regret is implied by the stability of queues that are fully utilized by the fluid relaxation, provided that only the matches active in the fluid relaxation are performed.
For the special case of two-way matching systems, \cite{mairesse2016stability} identify necessary and sufficient conditions for stability, and \cite{jonckheere2023generalized} show that several policies, including the longest-queue policy and a generalized max-weight policy, achieve the maximal stability region.
\cite{comte2021online} relate the stability region and achievable long-run matching rates to the polytope corresponding to a conservation equation.

\section{Model}\label{sec:model_setup}

We study a centralized dynamic matching market.
There are $n$ types of agents $\calA = [n]$ and $k$ types of matches $\calM = [k]$.
For each $m \in \calM$, the value of performing a match $m$ is denoted by $r_m > 0$.
Following \cite{kerimov2025optimality}, we focus on
two-way matching structures, i.e., each match $m\in\calM$ consists of exactly two agent types.
We encode the matching structure into a \emph{matching matrix} $M \in \{0, 1\}^{\calA \times \calM}$, where $M_{im} = 1$ if and only if {match} $m$ involves agent type $i$ for all $i \in \calA$ and $m \in \calM$.
Assume without loss of generality that each agent type participates in at least one match.
We write $(\calA, \calM)$ for the undirected graph whose vertices are agent types and whose edges are matches.

We consider discrete-time arrivals where exactly one agent arrives and joins
the type-dedicated queue at each integer time $t\geq1$.
Let $\theta_t\in\calA$ denote the type of the agent arriving at time $t$.
For simplicity, we assume that there are no preexisting agents in the system.\footnote{This is a standard assumption made by the  literature~\cite{wei2026constant,kerimov2025optimality}.}
Let $\lambda \in \reals_{> 0}^n$ denote a probability distribution over agent types, where $\lambda_i$ represents the probability of an {arriving} agent being type $i$, and $\sum_{i=1}^n \lambda_i = 1$.
Arrival types are independent across periods with common distribution $\lambda$.
For each type $i$, let $A_i(t)$ denote the cumulative number of type $i$ arrivals up to and including time $t$.
At each period $t$, the central decision maker can perform match $m \in \calM$ only if there are waiting agents of each type contained in $m$.
This match then generates a reward $r_m$, and the agents participating in the match depart.
For the sequence of events within each period, an agent first arrives, and then the decision maker decides which matches to perform.

A \emph{matching policy} decides how to perform matches at each period.
We will restrict our attention to non-anticipative policies, whose decisions at each moment only depend on the events that have occurred so far.
Given a policy, for every match $m \in \calM$ and time $t \geq 0$, we let $D_m(t)$ denote the number of matches of type $m$ performed by the policy up to and including time $t$.
We maintain a queue for each type $i \in \calA$, which contains all agents of type $i$ in the system that have not been matched.
Define $Q_i(t) \triangleq (A(t) - MD(t) - F(t))_i$ as the length of each queue $i \in \calA$ after time $t \geq 0$, where $F_i(t)$ is the number of type-$i$ agents discarded up to time $t$ (discarding is introduced at the end of this section, and $F\equiv0$ for policies that never discard), and
we refer to $Q(t)$ as the \emph{state} of the system after time $t$.

\subsection{Optimality criterion}\label{sec:all_time_regret}

The expected total value generated by a policy $\Pi$ during the first $t$ periods is denoted by $\calR^{\Pi}(t) \triangleq \Expect[r^T D(t)]$.
Let $\calR^*(t)$ be the expected value attained by the policy that takes no action before time $t$ and performs matches that maximize the overall rewards up to time $t$.
Formally,
\begin{align}\label{eq:def-hin-opt}
    \calR^*(t)
    \triangleq \Expect\left[\begin{array}{ll}\max & r^T y \\ \text { s.t. } & M y \leq A(t) \\ & y \in \mathbb{Z}_{\geq 0}^k\end{array}\right],
\end{align}
and $\calR^*(t)$ is an upper bound on $\calR^{\Pi}(t)$ for every policy $\Pi$.
For each integer $t\geq0$, define the \emph{regret} of a policy $\Pi$ after time $t$ as $\calR^*(t) - \calR^{\Pi}(t)$, and define the \emph{all-time regret} of $\Pi$ after time $t$ as
\begin{align*}
    \regret(\Pi, t) \triangleq \sup_{0 \leq u \leq t} (\calR^*(u) - \calR^{\Pi}(u)).
\end{align*}
We say that $\Pi$ achieves constant regret at all times (or all-time constant regret) if $\regret(\Pi, t)$ is bounded above by a constant that does not depend on $t$.
Note that this performance metric differs from simply achieving exact hindsight optimality at the end of time horizon $t$,
which can be attained by a trivial policy that delays all matches until the end of the time horizon $t$ and then solves an optimization problem to maximize overall rewards.
In other words, to achieve all-time constant regret, a policy must account for both short-term and long-term rewards, ensuring that decisions made at each step do not compromise overall performance.

\subsection{Static-planning and the general position gap}

We consider the fractional relaxation of the integer program in \eqref{eq:def-hin-opt}.
By Jensen's inequality,
\begin{align*}
    \calR^*(t)
    = \Expect\left[\begin{array}{ll}\max & r^T y \\ \text { s.t. } & M y \leq A(t) \\ & y \in \mathbb{Z}_{\geq 0}^k\end{array}\right]
    \leq \begin{array}{ll}\max & r^T x \\ \text { s.t. } & M x \leq t\lambda. \\ & x \in \mathbb{R}_{\geq 0}^k\end{array}
\end{align*}
Replacing $z = x / t$ and adding slack variables $(s_i)_{i \in \calA}$, the linear program in the RHS can be rewritten as
\begin{equation}\label{eq:static-planning}
    \SPP(\lambda) \triangleq
    \begin{array}{ll}\max & r^T z \\ \text { s.t. } & M z + s = \lambda \\ & z \in \mathbb{R}_{\geq 0}^k, s \in \reals_{\geq 0}^n\end{array},
\end{equation}
which we refer to as the \emph{static-planning problem}.
Here, $z_m$ represents the planned rate of match $m$, measured in matches per period. In an optimal solution $(z^*,s^*)$, the vector $z^*$ specifies the optimal matching rates.

Next, we introduce the notion of general position gap that captures the stability level of $\SPP(\lambda)$.
This condition is proved necessary for any policy to achieve constant regret at all times (see, e.g., \cite[Example 3.1]{kerimov2024dynamic}).
We note that, for a chosen nondegenerate optimal basis, our formula
for the general position gap coincides with that in
\cite{kerimov2024dynamic,kerimov2025optimality}.
Unlike those formulations, we require only the existence of a nondegenerate optimal basic solution and do not require the optimal
solution to be unique.
The gap parameter and active graph are defined relative to the chosen nondegenerate optimal basic solution.

\begin{definition}[General position gap]
    We say that the \emph{general position gap} (GPG) condition holds if $\SPP(\lambda)$ has a non-degenerate optimal basic solution, meaning that all $n$ basic variables in this solution are strictly positive.
    For such a solution $(z^*, s^*)$, define the GPG parameter $\epsilon$ as the minimum value of all basic variables, i.e.,
    \begin{align}\label{eq:gpg}
        \epsilon
        \triangleq \min\left\{\min_{m \in \calM\,, z_m^* > 0} z_m^*,\ \min_{i \in \calA\,, s_i^* > 0} s_i^*\right\}.
    \end{align}
\end{definition}

The GPG condition is a standard assumption in online revenue management and dynamic matching literature (see, e.g., \cite{jasin2012, DBLP:journals/ior/ChenL024, kerimov2024dynamic}), and any linear program can satisfy this condition with an arbitrarily small perturbation \cite{megiddo1989varepsilon}.

Under the GPG condition, let $(z^*, s^*)$ be a non-degenerate optimal basic solution of $\SPP(\lambda)$ with $\epsilon$ defined as \eqref{eq:gpg}.
Define $\calM_+ \triangleq \{m \in \calM \mid z_m^* > 0\}$ as the set of active matches, call the other matches redundant, and write $r_{\max}\triangleq\max\{r_m:m\in\calM_+\}$ for the largest reward among the active matches.
Also, define $\calA_+ \triangleq \{j \in \calA \mid s_j^* > 0\}$ and $\calA_0 \triangleq \calA \setminus \calA_+$ as the set of under-demanded queues and the set of over-demanded queues, respectively.

As an important consequence of the GPG condition, which is also commonly used by prior work (see, e.g., \cite[Lemma 5.1]{kerimov2025optimality} and \cite[Lemma 1]{gupta2024greedy}), bounding the all-time regret of a policy can be boiled down to analyzing the total length of the over-demanded queues, provided that the policy is restricted to active matches.

\begin{lemma}\label{lmm:opt-test}
    Suppose that the GPG condition holds, and let $(z^*, s^*)$ be a non-degenerate optimal basic solution of $\SPP(\lambda)$.
    Suppose that the following conditions hold under a policy $\Pi$:
    \begin{enumerate}
        \item Only matches in $\calM_+$ are performed and agents are discarded only at types in $\calA_+$, and
        \item $\sum_{i \in \calA_0} \Expect[Q_i(t)] \leq \bar B$ for every $t > 0$, where $\bar B > 0$ does not depend on $t$.
    \end{enumerate}
    Then $\regret(\Pi,t)\leq r_{\max}\bar B$ for every $t\geq0$.
\end{lemma}

The proof of \Cref{lmm:opt-test} is given in Appendix~\ref{sec:ctp-regret-proof}.
In view of \prettyref{lmm:opt-test}, to achieve constant regret at all times, it suffices to restrict the policies to only use active matches and control the lengths of over-demanded queues. Several existing policies restrict matching to the active matches of a chosen optimal basic solution~\cite{kerimov2024dynamic,gupta2024greedy,kerimov2025optimality}. 
The primal-dual policy in~\cite{wei2026constant} does not impose this restriction.
From now on, we restrict to active matches and write $\calM = \calM_+$.
Actions on one connected component of the active graph $(\calA, \calM_+)$ do not affect other connected components, so we analyze each component separately and sum the resulting queue bounds.
From now on, we assume that the active graph only contains one connected component, i.e., it is connected.

We say that a queue $i \in \calA$ is \emph{truncated} if, whenever a type-$i$ agent arrives and is not matched within the same period, this agent is immediately discarded, where discarding agents can be equivalently viewed as setting these agents aside and never matching them.
We assume that all queues in $\calA_+$ are truncated.\footnote{Under the GPG condition, without discarding, the expected length of every under-demanded queue will grow unbounded under any matching policy with bounded all-time regret~\cite{kerimov2024dynamic}. Hence, discarding is necessary to ensure the ergodicity of the Markov chain $(Q(t))_{t \geq 0}$.}
Under this assumption, we can still apply \Cref{lmm:opt-test} to analyze the all-time regret of policies that only use active matches.

\section{Main results}\label{sec:ctp-policy}

This section defines the General Priority policy, states the main results, and outlines their proofs.
We write $\Expect_q$ for expectation when $Q(0)=q$.

\subsection{Static priority policies and main results}

For the chosen non-degenerate optimal basic solution $(z^*,s^*)$, the active graph $(\calA,\calM_+)$ is either a tree with exactly one under-demanded node or an odd-unicyclic graph with no under-demanded node~\cite[Lemma~3.1]{kerimov2025optimality}.

When the active graph is a tree, the \emph{General Priority} (\textbf{GP}) policy roots the tree at its unique under-demanded node.
When the active graph is odd-unicyclic, let $\calA_{\mathrm{cyc}}\subseteq\calA$ denote the vertices of its unique cycle and orient the cycle clockwise; removing the cycle edges leaves attached trees, each rooted at a cycle vertex $i\in\calA_{\mathrm{cyc}}$.
Each arriving agent is matched to the first nonempty queue in the following priority list, if one exists:
\begin{itemize}
 \item at a non-root vertex: its children, in a fixed order, then its parent;
 \item at an under-demanded root: its children, in a fixed order, and the agent is discarded if all of them are empty;
 \item at a cycle vertex: its attached children, in a fixed order, then its clockwise neighbor, then its counterclockwise neighbor.
\end{itemize}
An agent that is neither matched nor discarded is retained in its queue.
On trees we also call this policy \emph{Tree Priority} (\textbf{TP}), and on odd-unicyclic graphs we refer to it as the \emph{Clockwise Priority} (\textbf{CP}) policy.
From an arbitrary initial vector, \textbf{GP} applies the same rule to each arrival and performs no additional matches between agents already present initially.

Thus \textbf{GP} is a \emph{static priority policy}: each agent type has a fixed priority order over its active neighbors.
Once the basis is fixed, \textbf{GP} decides using only which neighboring queues are nonempty: it uses neither queue lengths nor arrival rates, and it never re-solves the static-planning problem.

The total queue length of a state $q$ is $\lVert q\rVert_1=\sum_{i\in\calA}q_i$.
We first state the queue-length and regret bounds for the \textbf{TP} policy on trees.

\begin{theorem}[Tree Priority]\label{thm:ctp-tree}
Under the GPG condition, suppose that the active graph $(\calA,\calM_+)$ of the chosen basis is a tree.
Then, for any $t\geq0$, $\regret(\mathbf{TP},t)=O(\epsilon^{-1})$.
\end{theorem}

We remark that $\TP$ was first proposed by \cite{kerimov2025optimality}.
They only show that $\TP$ achieves bounded all-time regret (i.e., the Markov chain $(Q(u))_{u \geq 0}$ is ergodic), while \Cref{thm:ctp-tree} settles the optimal regret dependence of $\TP$ in terms of $\epsilon$.

We next consider the general case, where the active graph may also be odd-unicyclic.
The following theorem gives the corresponding bounds for the \textbf{GP} policy.

\begin{theorem}[General Priority]\label{thm:ctp-main}
Under the GPG condition, the set of states reachable from the empty state is a positive recurrent class of the queue-length chain under $\mathbf{GP}$. Moreover, for any $t\geq0$, $\regret(\mathbf{GP},t)=O(\epsilon^{-1})$.
\end{theorem}

Theorem~\ref{thm:ctp-tree} is proved at the end of Section~\ref{sec:ctp-prelim}, and Theorem~\ref{thm:ctp-main} at the end of Section~\ref{sec:ctp-completion}.
We remark that the regret bounds in both theorems depend polynomially on the network size $n$, and we do not optimize this dependence.

\subsection{Technical overview}\label{sec:ctp-overview}

\paragraph{Queue-length bound for trees.}
We begin with the queue-length bound for trees.
Recall that the tree is rooted at its unique under-demanded queue $r$.
To fix notation, for vertices $i$ and $j$, let $d(i,j)$ be the number of edges on the path between $i$ and $j$, and let $\sigma_{ij}:=(-1)^{d(i,j)+1}$.
Thus, $\sigma_{ij}=1$ when $j$ is an odd number of edges away from $i$, and $\sigma_{ij}=-1$ when it is an even number of edges away.
Let $\calT(i)$ be the set of vertices in the subtree of $i$ (including $i$), and $\calT^-(i):=\calT(i)\setminus\{i\}$ be the set of strict descendants of $i$.

For any queue-length vector $q\in\mathbb Z_{\geq0}^n$, define the alternating queue sum
\begin{equation}\label{eq:ctp-workload}
 Y_i(q):=\sum_{j\in\calT^-(i)}\sigma_{ij} q_j.
\end{equation}
For any $q$ with $q_r=0$, we have\footnote{To verify~\eqref{eq:ctp-workload-identities}, for each $j\ne r$ at depth $h:=d(r,j)$, the coefficient of $q_j$ on the right-hand side is $(-1)^{h+1}+2\sum_{k=1}^{h-1}(-1)^{k+1}=1$.}
\begin{equation}\label{eq:ctp-workload-identities}
 \norm{q}_1=Y_r(q)+2\sum_{i\in\calT^-(r)}Y_i(q).
\end{equation}
The root queue is truncated, so $Q_r(t)=0$. To bound $\Expect[\lVert Q(t)\rVert_1]$, it suffices to upper bound $\Expect[Y_i(Q(t))]$ for each queue $i$.
For brevity, we write $Y_i(t)=Y_i(Q(t))$.
As an advantage of working with $Y_i(t)$ instead of $Q(t)$ directly, when a match is performed, and both endpoints of this match belong to $\calT^-(i)$, the value of $Y_i$ remains the same by~\eqref{eq:ctp-workload}.
Consequently, to track $Y_i$, it suffices to track the matches between $i$ and its children, and the arrivals among the strict descendants of $i$.
In particular, when at least one child of $i$ is nonempty, the increment of $Y_i$ at time $t$ is given by the signed arrival increment
\begin{equation}\label{eq:ctp-signed-arrival}
 \xi_i(t):=\sum_{j\in\calT(i)}\sigma_{ij} \1{\theta_t=j}.
\end{equation}
Since only one agent arrives per period, $|\xi_i(t)|\leq1$.
The child-first service rule of \textbf{TP} then yields the following pathwise recursive upper bound for $Y_i(t)$, where the second term in the RHS comes from the case where all children of $i$ are empty:
\begin{align}
 Y_i(t)\leq\max\left\{Y_i(t-1)+\xi_i(t),\ \sum_{j\in\calG(i)}Y_j(t)\right\},
 \label{eqn:cur-ub-yit}
\end{align}
where $\calG(i)$ denotes the set of grandchildren of $i$.

The recursion \eqref{eqn:cur-ub-yit} admits the following interpretation: the bound either moves backward by one period and adds a negative-drift increment, or descends two levels in the tree.
After taking the latter choice, we can apply the same recursive arguments to the grandchildren.
This recursion ends when it reaches time zero, giving an upper bound for $Y_i(t)$ expressed as the sum of bounded negative-drift increments.
Combining this recursion with the GPG condition, we establish the pathwise upper bound
\begin{align}
    Y_i(t)
    \leq \sum_{j \in \calT(i)} \max_{0 \leq \ell \leq t} \left\{ \sigma_{ij} \sum_{k=\ell + 1}^t \bigl(\1{\theta_k=j}-\lambda_j\bigr) - \frac{\epsilon}{n} (t - \ell) \right\},
    \label{eqn:ub-yit-ref-random-walk}
\end{align}
and Kingman's inequality~\cite{kingman1962inequalities} bounds the mean of each reflected random walk (i.e., each term in the above summation) by $O(n/\epsilon)$.
Hence $\Expect Y_i(t)=O(n^2/\epsilon)$ for every $i$, and~\eqref{eq:ctp-workload-identities} gives $\sup_{t\geq0}\Expect_0\lVert Q(t)\rVert_1=O(n^3/\epsilon)$.

\paragraph{Queue-length bound for odd-unicyclic graphs.}

Recall that an odd-unicyclic active graph consists of an odd cycle with trees attached to its vertices.
Deleting the cycle edges for analysis leaves disjoint rooted trees $\calT_i$, one at each cycle vertex $i\in\calA_{\mathrm{cyc}}$.
We define $Y_i$ as in~\eqref{eq:ctp-workload}, using descendants and distances within these attached trees; at a cycle root $i$, the sum ranges over $\calT_i\setminus\{i\}$.

Because attached children have priority over cycle neighbors, the tree bound~\eqref{eqn:ub-yit-ref-random-walk} remains valid.
When all queues outside the cycle start empty, this bound and the definition in~\eqref{eq:ctp-workload} give the following pathwise estimate; see Lemma~\ref{lem:ctp-branches}:
\begin{align*}
    \sum_{i\in\calA_{\mathrm{cyc}}} \max_{0 \leq \ell \leq t} |Y_i(\ell)|
    = O(R(t)),
\end{align*}
where the term
\begin{equation}\label{eq:ctp-arrival-range}
 R(t):=\sum_{j\in\calA}\max_{0\leq \ell\leq \ell'\leq t}
 \left|A_j(\ell')-A_j(\ell)-\lambda_j(\ell'-\ell)\right|
\end{equation}
controls the arrival-fluctuation error up to time $t$, and has first and second moments $O(\sqrt{nt})$ and $O(nt)$, respectively.
Using the same lemma, we can bound the total queue length outside the cycle by $O(nR(t))$ throughout $[0,t]$.

In the remainder of this subsection, we focus on the cycle vertices.
As the most crucial step of our analysis, we isolate the properties of the cycle dynamics by considering an auxiliary process starting with a singleton state $\mu\mathbf e_s$ for some source vertex $s$ on the cycle.
Here, $\mu\in\mathbb Z_{\geq1}$ is the initial queue length at $s$.
The nonexpansiveness property of the queue-length process under the \textbf{GP} policy then, via coupling, enables us to transfer the drift property of this auxiliary process to the original process.

Suppose that the cycle length is $2m+3$, and we identify cycle vertices by $\mathbb Z_{2m+3}$ with $i+1$ going clockwise.
Notice that, whenever $s$ remains nonempty, the match $\{s - 2, s - 1\}$ is never used: every arrival to queue $s - 1$ will be matched to either a child of $s - 1$ or the source $s$, and hence queue $s - 1$ stays empty.
Consequently, deleting the unused edge $\{s-2,s-1\}$ opens the cycle into a path.
Starting from $s-2$ and proceeding counterclockwise, label alternating vertices $r_k=s-2k$ for $1\leq k\leq m+1$ and $u_k=s-(2k+1)$ for $1\leq k\leq m$, with all indices taken modulo $2m+3$.
With these labels, the path is
\[
 r_1-u_1-\cdots-r_m-u_m-r_{m+1}-s-(s-1).
\]
Using the optimal matching-rate vector $z^*$ from the chosen basic solution of the static-planning problem~\eqref{eq:static-planning}, write $a_0 := z^*_{\{s-2,s-1\}}$ for the LP matching rate associated with the cut edge and $a_k := z^*_{\{u_k,r_{k+1}\}}$ for the rates associated with the edges $\{u_k,r_{k+1}\}$, $1\leq k\leq m$, as depicted below:
\[
 \begin{tikzpicture}[x=1.25cm,y=1cm,color=.,every node/.style={inner sep=2pt}]
  \node[queuevertex,fill=rnodecolor,label={[text=black]below:$r_1$}] (r1) at (0,0) {};
  \node[queuevertex,fill=unodecolor,label={[text=black]below:$u_1$}] (u1) at (1,0) {};
  \node[queuevertex,fill=rnodecolor,label={[text=black]below:$r_2$}] (r2) at (2.1,0) {};
  \node (dots) at (3.3,0) {$\cdots$};
  \node[queuevertex,fill=unodecolor,label={[text=black]below:$u_m$}] (um) at (4.5,0) {};
  \node[queuevertex,fill=rnodecolor,label={[text=black]below:$r_{m+1}$}] (rnext) at (5.8,0) {};
  \node[queuevertex,fill=white,label={[text=black]below:$s$}] (source) at (7,0) {};
  \node[queuevertex,fill=white,label={[text=black]below:$(s-1)$}] (last) at (8.3,0) {};
  \draw (r1)--(u1);
  \draw (u1)--node[above] {$a_1$} (r2);
  \draw (r2)--(dots)--(um);
  \draw (um)--node[above] {$a_m$} (rnext);
  \draw (rnext)--(source)--(last);
  \draw[dashed] (r1.north) .. controls (0,1.5) and (8.3,1.5) ..
   node[midway,above] {$a_0\ \text{(unused)}$} (last.north);
 \end{tikzpicture}
\]
It is instrumental to first inspect the queue-length process under fluid arrivals.
Denote
\begin{align*}
    \underline a_k&=\min_{0\leq j\leq k}a_j\quad(0\leq k\leq m)
\end{align*}
as the running minima of the LP matching rates, with $\underline a_0=a_0$.
Under fluid arrivals, an induction along the path, carried out for the stochastic process in Lemma~\ref{fact:ctp-queue-profile}, shows that each queue $r_k$ for $k \in [m]$ grows at rate $\underline a_{k-1}-\underline a_k$, the source queue $s$ shrinks at rate $\underline a_0 + \underline a_m$, and the remaining queues stay empty.
Hence, the growth rate of the total queue length is
\begin{align*}
    \sum_{k=1}^m (\underline a_{k-1}-\underline a_k) - (\underline a_0 + \underline a_m)
    = -2 \underline a_m
    \leq -2\epsilon.
\end{align*}
The inequality follows from~\eqref{eq:gpg}: each $a_j$ is a positive basic matching rate, so $a_j\geq\epsilon$ and hence $\underline a_m\geq\epsilon$.

For the fluid process started from $\mu\mathbf e_s$ (initially $\mu$ agents at $s$ and none elsewhere), consider the regime $a_0=\Theta(1)$ and $\underline a_m=\Theta(\epsilon)$. The source then empties after $\mu / (\underline a_0 + \underline a_m) = \Theta(\mu)$ periods, and the total queue length decreases by $\Theta(\epsilon\mu)$.
Arrival fluctuations over these $\Theta(\mu)$ periods are of order $\sqrt{\mu}$.
Hence, making the total decrease comparable to the fluctuation error requires $\mu = \Omega(\epsilon^{-2})$, yielding a suboptimal queue-length bound $O(\epsilon^{-2})$.

We now describe how to overcome the above barrier.
Denote
\[
 \calA_s^{\mathrm{ret}}=\{r_k:1\leq k\leq m,\ \underline a_k<\underline a_{k-1}\}
\]
as the set of queues with a strictly positive growth rate under fluid arrivals.
When the source empties, and the total queue length remains significant, the queues in $\calA_s^{\mathrm{ret}}$ must have accumulated a significant amount of mass.
We show that the same growth rates under fluid arrivals are preserved as long as the source is nonempty and all other nonempty queues belong to $\calA_s^{\mathrm{ret}}$.
As a result, after the source empties, we can designate the nonempty queue in $\calA_s^{\mathrm{ret}}$ with the smallest index as the new source, and the same argument implies that the total queue length again decreases at rate at least $2\epsilon$.
After each switch, the retained set $\calA_s^{\mathrm{ret}}$ is recomputed for the new source and may include vertices outside the original retained set.
For source $s$, the cut edge is $\{s-2,s-1\}$, with LP matching rate $a_0$.
When the source changes to $s'=r_k\in\calA_s^{\mathrm{ret}}$, the new cut edge is $\{s'-2,s'-1\}=\{r_{k+1},u_k\}$, whose rate satisfies $a_k=\underline a_k<\underline a_{k-1}\leq a_0$.
Indeed, $\underline a_k=\min\{\underline a_{k-1},a_k\}$, and $r_k\in\calA_s^{\mathrm{ret}}$ implies $\underline a_k<\underline a_{k-1}$, so $\underline a_k=a_k$.
Since the entire cycle has at most $n$ edges and each has a fixed LP matching rate, this strict decrease prevents revisiting a cut edge, allowing at most $n-1$ source switches.
Now, it takes $O(\mu / \epsilon)$ periods to exhaust all the initial mass $\mu$.
Choosing $\mu = \Omega(\epsilon^{-1})$ is sufficient for the total decrease to subsume the fluctuation error $O(\sqrt{\mu / \epsilon})$, yielding the optimal queue-length bound $O(\epsilon^{-1})$ for odd-unicyclic graphs.

We call this auxiliary evolution up to its first empty state a \emph{cycle-draining phase}.
Each cycle-draining phase then consists of several \emph{source-draining stages}, with each stage draining one source queue.
Under discrete arrivals, fluctuations create small queues outside the retained set.
We control these queues by projecting an auxiliary process at source switches and use nonexpansive coupling to bound its distance from the original process.
Section~\ref{sec:ctp-proof} gives the formal construction and derives the queue-length bound through a deterministic-horizon comparison and quadratic drift.

\section{Analysis}\label{sec:analysis}\label{sec:prelim-notation}

Section~\ref{sec:prelim-tools} collects three tools: a coupling property of \textbf{GP}, bounds on arrival fluctuations, and reflected random walks.
Section~\ref{sec:ctp-prelim} treats trees and proves Theorem~\ref{thm:ctp-tree}, and Section~\ref{sec:ctp-proof} treats odd-unicyclic graphs and proves Theorem~\ref{thm:ctp-main}.
We first fix notation used throughout.

Recall that $\calT(i)$ is the subtree rooted at $i$ and that $\calT^-(i)=\calT(i)\setminus\{i\}$ is the set of its strict descendants.
On odd-unicyclic graphs, this notation applies within each attached tree, excluding cycle edges.
For a non-root vertex $i$, write $P(i)$ for its parent.
For every vertex $i$, let $\calC(i)$ be the set of its children and $\calG(i)=\bigcup_{j\in\calC(i)}\calC(j)$ its grandchildren.
Finally, $Y_i(q)=\sum_{j\in\calT^-(i)}\sigma_{ij}q_j$ denotes the alternating queue sum~\eqref{eq:ctp-workload}, and along the process we write $Y_i(t)=Y_i(Q(t))$.

\subsection{Preliminaries}\label{sec:prelim-tools}\label{sec:prelim-bounds}
\label{sec:prelim-coupling}\label{sec:prelim-fluctuations}\label{sec:prelim-reflected-walks}

We call a queue vector $q$ an \emph{independent-support state} if $q_iq_j=0$ on every active edge $\{i,j\}$.
\textbf{GP} preserves this property at every arrival, so the empty-start process remains in independent-support states.

We first introduce the following lemma that establishes the nonexpansiveness property of the queue-length process.
We will use this lemma to transfer estimates for singleton-start processes to general-start processes, and estimates for the stationary state to the empty-start process.

\begin{lemma}\label{lem:prelim-coupling}
Suppose $Q(0)=q$ and $Q'(0)=q'$, with $q,q'\in\mathbb Z_{\geq0}^{\calA}$.
If $Q$ and $Q'$ follow $\mathbf{GP}$ with the same fixed priority orders and the same arrivals, then
\begin{equation}\label{eq:prelim-coupling}
 \|Q(t)-Q'(t)\|_1\leq\|q-q'\|_1\qquad(t\geq0).
\end{equation}
If $\pi$ is an invariant distribution of $\mathbf{GP}$ with finite first moment, then
\begin{equation}\label{eq:prelim-stationary-empty}
 \sup_{t\geq0}\Expect_0\norm{Q(t)}_1\leq2\Expect_\pi\norm{Q}_1,
\end{equation}
where $\Expect_\pi$ denotes expectation for the chain started from $\pi$.
\end{lemma}

Define
\begin{equation*}
 M_j(t)=A_j(t)-\lambda_jt,\qquad
 R_j(t)=\max_{0\leq a\leq b\leq t}|M_j(b)-M_j(a)|,
\end{equation*}
where $R(t)=\sum_{j\in\calA}R_j(t)$ is the fluctuation-error term defined in~\eqref{eq:ctp-arrival-range}.
Note that $M_j(t)$ is the centered arrival count up to time $t$.
The next lemma bounds the fluctuations from mean arrivals over every subinterval of a fixed observation horizon.
\begin{lemma}\label{lem:prelim-fluctuations}
For every deterministic nonnegative integer $t$,
\begin{align}
 \Expect R(t)&=O(\sqrt{nt}),\label{eq:prelim-range-first-moment}\\
 \Expect[R(t)^2]&=O(nt).\label{eq:prelim-range-second-moment}
\end{align}
Pathwise, $R_j(t)$ and $R(t)$ are nondecreasing in $t$, and $|M_j(b)-M_j(a)|\leq R_j(t)$ whenever $0\leq a\leq b\leq t$.
\end{lemma}

The pathwise guarantees in \Cref{lem:prelim-fluctuations} will be applied to the process restarted at a possibly random time $S$.
In this case, we use $R^{(S)}(U)$ to denote the error term~\eqref{eq:ctp-arrival-range} computed from the shifted arrivals $A_j(S+\cdot)-A_j(S)$.
Its centered increments are increments $M_j(S+b)-M_j(S+a)$ of the original arrivals, so Lemma~\ref{lem:prelim-fluctuations} gives, pathwise,
\begin{equation}\label{eq:ctp-range-shift}
 R^{(S)}(U)\leq R(t)\qquad\text{whenever }S+U\leq t.
\end{equation}

Fix $\delta>0$.
For a vertex $j$ and a fixed sign $\sigma\in\{-1,1\}$, define the reflected walk starting from $H_j^\sigma(0)=0$ by
\begin{equation}\label{eq:prelim-reflected-walk}
 H_j^\sigma(t)=\left[H_j^\sigma(t-1)
 +\sigma\bigl(M_j(t)-M_j(t-1)\bigr)-\delta\right]^+
 \qquad(t\geq1).
\end{equation}
The following lemma gives both a pathwise bound in terms of the arrival range and a uniform first-moment bound for the reflected walks.

\begin{lemma}\label{lem:prelim-reflected-walk}
The reflected walks defined in \eqref{eq:prelim-reflected-walk} satisfy
\begin{equation}\label{eq:prelim-reflected-maximum}
 H_j^\sigma(t)=\max_{0\leq \ell\leq t}
 \left\{\sigma\bigl[M_j(t)-M_j(\ell)\bigr]-\delta(t-\ell)\right\},
\end{equation}
and
\[
 \sup_{0\leq\ell\leq t}H_j^\sigma(\ell)\leq R_j(t),
 \qquad \sup_{t\geq0}\Expect H_j^\sigma(t)=O(1/\delta).
\]
\end{lemma}

The proofs of Lemmas~\ref{lem:prelim-coupling}--\ref{lem:prelim-reflected-walk} are deferred to Appendix~\ref{sec:prelim-proofs}.

\subsection{Analysis of the Tree Priority policy}\label{sec:ctp-prelim}\label{sec:ctp-alternating-sums}

We prove \Cref{thm:ctp-tree} in this subsection.

Note that the signed arrival increment $\xi_i(t)$ in~\eqref{eq:ctp-signed-arrival} has mean $-\epsilon_i$, where
\[
 \epsilon_i:=\sum_{j\in\calT(i)}(-1)^{d(i,j)}\lambda_j
 =-\Expect\xi_i(t).
\]
Substituting $\lambda=Mz^*+s^*$ into the alternating sum in the definition of $\epsilon_i$ and canceling the contributions of internal tree edges gives
\begin{equation}\label{eq:ctp-drift-parameters}
 \epsilon_i=
 \begin{cases}
  z^*_{\{i,P(i)\}}, & i\text{ is a non-root vertex},\\
  s_i^*, & i\text{ is the root of a tree component},\\
  \displaystyle\sum_{j:\,\{i,j\}\text{ cycle edge}}z^*_{\{i,j\}},
     & i\in\calA_{\mathrm{cyc}},
 \end{cases}
\end{equation}
where, in the last case, $i$ is the root of its attached tree.
Then, \eqref{eq:gpg} implies $\epsilon_i\geq\epsilon$.

In the following proposition, we bound the alternating queue sums $Y_i(t)$ by reflected random walks defined in \eqref{eq:prelim-reflected-walk}.

\begin{proposition}\label{lem:ctp-timestamps}
Under child-first priorities, if the strict subtree of a vertex $i$ starts empty, then for every $t\geq0$,
\begin{equation}\label{eq:ctp-timestamps}
 Y_i(t)\leq\sum_{j\in\calT(i)}H_j^{\sigma_{ij}}(t),
\end{equation}
where $H_j^{\sigma_{ij}}$ is defined by~\eqref{eq:prelim-reflected-walk} with $\delta=\epsilon/n$.
\end{proposition}

\begin{proof}
We first establish the recursion
\begin{equation}\label{eq:ctp-recursion}
 Y_i(t)\leq\max\left\{Y_i(t-1)+\xi_i(t),\ \sum_{j\in\calG(i)}Y_j(t)\right\}
 \qquad(t\geq1),
\end{equation}
which holds for every vertex $i$ under child-first priorities, regardless of the initial state.
Cancellation of internal matches gives the exact increment identity
\[
 Y_i(t)-Y_i(t-1)=\xi_i(t)+\1{\theta_t=i}
 -\#\{\text{matches between }i\text{ and its children at period }t\}.
\]
If some child of $i$ is occupied before the arrival, an arrival at $i$ makes one such match by child-first priority, so $Y_i(t)\leq Y_i(t-1)+\xi_i(t)$ for every arrival type, including for initial states without independent support.
If all children are empty before the arrival but some child becomes occupied, the arrival is retained at that child, so $Y_i(t)=Y_i(t-1)+1=Y_i(t-1)+\xi_i(t)$.
In the remaining case, all children are empty after the arrival.
Grouping the sum in~\eqref{eq:ctp-workload} by child subtrees gives the identity $Y_i(q)=\sum_{j\in\calC(i)}\bigl(q_j-Y_j(q)\bigr)$.
Applying this identity twice gives
\[
 Y_i(t)=-\sum_{j\in\calG(i)}Q_j(t)+\sum_{j\in\calG(i)}Y_j(t)
 \leq\sum_{j\in\calG(i)}Y_j(t).
\]
Thus one of the two terms in~\eqref{eq:ctp-recursion} bounds $Y_i(t)$ in every case.
This argument remains valid at the root of an attached tree because its children have priority over both neighbors on the cycle.

To derive the reflected-walk bound, put $S_i(t):=\sum_{h=1}^t\xi_i(h)$, with $S_i(0)=0$, and unroll~\eqref{eq:ctp-recursion} over time:
\[
 Y_i(t)\leq\max_{0\leq u\leq t}
 \left\{\sum_{j\in\calG(i)}Y_j(u)+S_i(t)-S_i(u)\right\}.
\]
The term $u=0$ includes the initial condition because the strict subtree starts empty, so $Y_i(0)=Y_j(0)=0$ for every grandchild $j$.
Each term retains the imbalance in the grandchildren's subtrees at time $u$ and adds the signed arrivals between $u$ and $t$.
The drift bound following~\eqref{eq:ctp-drift-parameters} gives $\epsilon_i\geq\epsilon\geq|\calT(i)|\delta$, so
\[
 \begin{aligned}
 S_i(t)-S_i(u)
 &=\sum_{j\in\calT(i)}\sigma_{ij}\bigl[M_j(t)-M_j(u)\bigr]
   -\epsilon_i(t-u)\\
 &\leq\sum_{j\in\calT(i)}
 \left\{\sigma_{ij}\bigl[M_j(t)-M_j(u)\bigr]-\delta(t-u)\right\}.
 \end{aligned}
\]
The maximum representation in Lemma~\ref{lem:prelim-reflected-walk} gives, for $0\leq u\leq t$,
\[
 H_j^\sigma(u)+\sigma\bigl[M_j(t)-M_j(u)\bigr]-\delta(t-u)
 \leq H_j^\sigma(t).
\]
Thus a fluctuation accumulated up to $u$, followed by the increments from $u$ to $t$, is bounded by the same reflected walk at time $t$.

We now induct on subtree height, with the leaf case following from $Y_i\equiv0$ and the nonnegativity of the reflected walks.
Assume~\eqref{eq:ctp-timestamps} holds for every grandchild $j$ at all times.
Each grandchild is two levels below $i$, so its descendants have the same distance parity relative to both roots.
The signs of its reflected walks therefore agree with those in the sum on the right-hand side of~\eqref{eq:ctp-timestamps}.
Since the grandchildren's subtrees are disjoint and the walks are nonnegative,
\[
 \sum_{j\in\calG(i)}Y_j(u)
 \leq\sum_{j\in\calT(i)}H_j^{\sigma_{ij}}(u).
\]
Combining the preceding bounds for each fixed $u\leq t$ yields
\begin{equation*}
 \begin{aligned}
 \sum_{j\in\calG(i)}Y_j(u)+S_i(t)-S_i(u)
 &\leq\sum_{j\in\calT(i)}\Bigl\{H_j^{\sigma_{ij}}(u)
   +\sigma_{ij}\bigl[M_j(t)-M_j(u)\bigr]-\delta(t-u)\Bigr\}\\
 &\quad\leq\sum_{j\in\calT(i)}H_j^{\sigma_{ij}}(t).
 \end{aligned}
\end{equation*}
Taking the maximum over $u$ proves~\eqref{eq:ctp-timestamps}.
\end{proof}

With the reflected-walk bound in Proposition~\ref{lem:ctp-timestamps}, we can now prove the following queue-length estimate for trees.

\begin{corollary}\label{lem:ctp-tree-bound}
Under the GPG condition, if the active graph $(\calA,\calM_+)$ of the chosen basis is a tree, then under $\mathbf{TP}$ we have
\[
 \sup_{t\geq0}\Expect_0\norm{Q(t)}_1=O(n^3/\epsilon).
\]
\end{corollary}

\begin{proof}
The process starts empty, so~\eqref{eq:ctp-timestamps} holds at every vertex.
Together with Lemma~\ref{lem:prelim-reflected-walk}, this gives, for every vertex $i$,
\[
 \sup_{t\geq0}\Expect_0 Y_i(t)
 \leq\sum_{j\in\calT(i)}\sup_{t\geq0}\Expect H_j^{\sigma_{ij}}(t)
 =O(n^2/\epsilon).
\]
To convert these alternating-sum bounds into a queue-length bound, apply the total-inventory identity~\eqref{eq:ctp-workload-identities} to the subtree rooted at any vertex $i$, with its root coordinate set to zero, to obtain
\begin{equation}\label{eq:ctp-subtree-inventory}
 \sum_{j\in\calT^-(i)}q_j=Y_i(q)+2\sum_{j\in\calT^-(i)}Y_j(q).
\end{equation}
At the truncated root $r$, we have $Q_r(t)=0$, so taking $i=r$ in~\eqref{eq:ctp-subtree-inventory} gives
\[
 \norm{Q(t)}_1=Y_r(t)+2\sum_{j\in\calT^-(r)}Y_j(t).
\]
Taking expectations and summing over the $n$ vertices proves the claimed $O(n^3/\epsilon)$ bound.
\end{proof}

We are ready to complete the proof of Theorem~\ref{thm:ctp-tree}.

\begin{proof}[Proof of Theorem~\ref{thm:ctp-tree}]
Corollary~\ref{lem:ctp-tree-bound} gives $\sup_{u\geq0}\Expect_0\lVert Q(u)\rVert_1=O(n^3/\epsilon)$.
The policy $\mathbf{TP}$ uses only active matches and discards agents only at the root, which is under-demanded.
Lemma~\ref{lmm:opt-test}, applied with $\bar B=\sup_{u\geq0}\Expect_0\lVert Q(u)\rVert_1$, therefore gives $\regret(\mathbf{TP},t)=O(r_{\max}n^3/\epsilon)$ for every $t\geq0$.
\end{proof}

\subsection{Analysis of the General Priority policy}\label{sec:ctp-proof}\label{sec:ctp-branches}

To prove Theorem~\ref{thm:ctp-main}, it remains to analyze odd-unicyclic graphs.

Write the total length of queues outside the cycle as
\begin{equation*}
 B(q):=\sum_{j\in\calA\setminus\calA_{\mathrm{cyc}}}q_j.
\end{equation*}
In the following lemma, we apply the reflected-walk comparison and the tree inventory identities established in \Cref{sec:ctp-alternating-sums} to give the pathwise and first-moment bounds for the non-cycle queues.
\begin{lemma}\label{lem:ctp-branches}
If all non-cycle queues start empty at time zero, meaning $Q_j(0)=0$ for every $j\in\calA\setminus\calA_{\mathrm{cyc}}$, then, pathwise for every finite $U$,
\begin{align}
 \sum_{i\in\calA_{\mathrm{cyc}}}\sup_{0\leq t\leq U}|Y_i(Q(t))|&=O\!\left(R(U)\right),\label{eq:ctp-branch-imbalance}\\
 \sup_{0\leq t\leq U}B(Q(t))&=O\!\left(nR(U)\right).\label{eq:ctp-branch-mass}
\end{align}
For the empty-start process, the first moment is uniformly bounded at deterministic times:
\begin{equation}\label{eq:ctp-refreshed-branches}
 \sup_{t\geq0}\Expect_0 B(Q(t))
 =O(n^3/\epsilon).
\end{equation}
\end{lemma}

The proof of Lemma~\ref{lem:ctp-branches} is deferred to Appendix~\ref{sec:proof-branches}.
Next, we define source-draining stages, which are used to form cycle-draining phases.

\subsubsection{Source-draining stages}\label{sec:ctp-source}

Fix a cycle vertex $s$ as source, and write the cycle length as $2m+3$, with $m\geq0$ and all cycle indices taken modulo $2m+3$.
For the opened path in Section~\ref{sec:ctp-overview}, set $r_k=s-2k$ for $1\leq k\leq m+1$ and $u_k=s-(2k+1)$ for $1\leq k\leq m$, so $r_{m+1}=s+1$; we call $\{r_k,u_k\}$, $1\leq k\leq m$, the $k$-th \emph{pair}.
Throughout Section~\ref{sec:ctp-proof}, we use the shorthand $z_i:=z^*_{\{i,i+1\}}$ for a cycle-edge flow.
The source-dependent boundary flows and their running minima satisfy
\begin{align}
 a_0&=z_{s-2},\label{eq:ctp-cut-flow}\\
 a_k&=z_{r_{k+1}}=z_{r_k-2}\quad(1\leq k\leq m),\label{eq:ctp-path-flows}\\
 \underline a_k&=\min_{0\leq j\leq k}a_j\quad(0\leq k\leq m).\label{eq:ctp-running-minima}
\end{align}
In particular, $\underline a_0=a_0$.
For a triangle, $m=0$, the pairs and retained set are empty, and $\underline a_m=a_0$.
Each $a_k$ is a positive basic variable, so $\underline a_m\geq\epsilon$ by~\eqref{eq:gpg}.
Recall the retained set $\calA_s^{\mathrm{ret}}=\{r_k:1\leq k\leq m,\ \underline a_k<\underline a_{k-1}\}$ from Section~\ref{sec:ctp-overview}.
Thus $r_k\in\calA_s^{\mathrm{ret}}$ if and only if $a_k<a_j$ for every $j<k$, in which case $\underline a_k=a_k$; in particular, the values $a_k$ strictly decrease along the retained set.
All path indices, boundary flows, running minima, and retained sets are recomputed when the source changes.

We say that a state $q$ is \emph{compatible} with $s$ if $q_s>0$ and $q_i=0$ for every $i\notin\{s\}\cup\calA_s^{\mathrm{ret}}$; in particular, all non-cycle queues are empty.
The source and retained vertices form an independent set, so a compatible state has independent support.
Starting from a compatible state, apply the $\CP$ policy and set
\[
 \tau_s=\inf\{t\in\mathbb Z_{\geq1}:Q_s(t)=0\}
\]
as the first time the source empties, with the convention that the infimum of the empty set is $\infty$.
We call the evolution up to $\tau_s$ a \emph{source-draining stage}.
Throughout a stage, both cycle neighbors of $s$ stay empty and the cut edge $\{s-2,s-1\}$ carries no match: the neighbors are empty initially by compatibility, and while $Q_s>0$ an arrival at $s-1$ is matched to an attached child or to its clockwise neighbor $s$, an arrival at $s+1$ is matched to an attached child, to $s+2$, or to $s$, and an arrival at $s-2$ finds $s-1$ empty.

After the source empties, we designate a new source and start a new source-draining stage.
To retain the fluid growth rates after the transition, we need to empty all queues not in $\calA_s^{\mathrm{ret}}$, so that the new starting state is compatible with the new source.
To this end, we apply the \emph{endpoint projection} of $s$, defined by
\[
 (\Pi_s q)_i=\begin{cases}q_i,&i\in\calA_s^{\mathrm{ret}},\\0,&\text{otherwise.}\end{cases}
\]
At a finite endpoint $\tau_s$ with $\Pi_sQ(\tau_s)\ne \zeros$, we then choose the next source to be $s'=r_{k_1}$, where $k_1=\min\{k:(\Pi_sQ(\tau_s))_{r_k}>0\}$; equivalently, $s'$ maximizes the flow $z_{i-2}$ on the cut edge $\{i-2,i-1\}$ of $i$ over the nonempty vertices $i\in\calA_s^{\mathrm{ret}}$.

For an initial state $q$ compatible with $s$, define the deterministic fluid approximation $\bar q(t)$ by:
\begin{equation}\label{eq:ctp-comparison-profile}
 \bar q_i(t)=
 \begin{cases}
  q_s-(a_0+\underline a_m)t, & i=s,\\
  q_{r_k}+(\underline a_{k-1}-\underline a_k)t, & i=r_k,\quad 1\leq k\leq m,\\
  0, & \text{otherwise.}
 \end{cases}
\end{equation}
The fluid approximation~\eqref{eq:ctp-comparison-profile} is the fluid evolution of a stage described in Section~\ref{sec:ctp-overview}: under mean arrivals the queues $u_k$ stay empty, each $r_k$ grows at the rate $\underline a_{k-1}-\underline a_k$, and the source drains at rate $a_0+\underline a_m$.
After its source coordinate reaches zero, the same affine formula is continued solely for comparison, so its source coordinate may be negative.
The increments in different coordinates of the fluid approximation telescope, so
\begin{equation}\label{eq:ctp-profile-mass}
 \sum_{i\in\calA}\bar q_i(t)=\norm{q}_1-2\underline a_m t\leq\norm{q}_1-2\epsilon t.
\end{equation}
Moreover, $\bar q_i\equiv0$ for every $i\notin\{s\}\cup\calA_s^{\mathrm{ret}}$, because $\underline a_k=\underline a_{k-1}$ and $q_{r_k}=0$ whenever $r_k\notin\calA_s^{\mathrm{ret}}$.

The following lemma bounds the error of this fluid approximation in terms of arrival fluctuations.
The proof of \Cref{fact:ctp-queue-profile} can be found in Appendix~\ref{sec:ctp-profile-proof}.

\begin{lemma}\label{fact:ctp-queue-profile}
For an initial state compatible with source $s$ and every finite $U\leq\tau_s$,
\begin{equation*}
 \sup_{0\leq t\leq U}\|Q(t)-\bar q(t)\|_1
 =O\!\left(n^2R(U)\right).
\end{equation*}
\end{lemma}

The following proposition gives the drift property and projection cost of source-draining stages; its first two parts follow directly from Lemma~\ref{fact:ctp-queue-profile} and~\eqref{eq:ctp-profile-mass}.

\begin{proposition}[Source-draining stage]\label{lem:ctp-source}
For an initial state $q$ compatible with source $s$ and every finite $t\leq\tau_s$, pathwise,
\begin{equation}\label{eq:ctp-source-drift}
 \norm{Q(t)}_1+2\epsilon t\leq \norm{q}_1+O\!\left(n^2R(t)\right).
\end{equation}
On the event $\tau_s<\infty$,
\begin{equation}\label{eq:ctp-projection}
 \|Q(\tau_s)-\Pi_sQ(\tau_s)\|_1=O\!\left(  n^2R(\tau_s)\right);
\end{equation}
moreover, whenever $\Pi_sQ(\tau_s)\ne \zeros$, the projected state $\Pi_sQ(\tau_s)$ is compatible with the new source $s'$ and satisfies $z_{s'-2}<z_{s-2}$; on a triangle, $\Pi_sQ(\tau_s)= \zeros$.
\end{proposition}

\begin{proof}
Adding the error bound for $\|Q(t)-\bar q(t)\|_1$ from Lemma~\ref{fact:ctp-queue-profile} to~\eqref{eq:ctp-profile-mass} proves~\eqref{eq:ctp-source-drift}.
At a finite depletion time $\tau_s$, the source coordinate $Q_s(\tau_s)$ is zero, and every other coordinate $i$ deleted by $\Pi_s$ satisfies $\bar q_i(\tau_s)=0$.
Hence,
\[
 \|Q(\tau_s)-\Pi_sQ(\tau_s)\|_1
 \leq\|Q(\tau_s)-\bar q(\tau_s)\|_1.
\]
Lemma~\ref{fact:ctp-queue-profile} therefore also proves~\eqref{eq:ctp-projection}.

Suppose now that $\Pi_sQ(\tau_s)\ne\zeros$, and let $s'=r_{k_1}$ be the next source.
Relative to $s'$, the path labels shift by $k_1$, so $r'_j=r_{k_1+j}$ and the new boundary flows are $a'_j=a_{k_1+j}$ for $0\leq j\leq m-k_1$.
Every other strictly positive coordinate of $\Pi_sQ(\tau_s)$ corresponds to some $r_k\in\calA_s^{\mathrm{ret}}$ with $k_1<k\leq m$; it becomes $r'_{k-k_1}$ and satisfies
\[
 a'_{k-k_1}=a_k<\min_{j<k}a_j\leq\min_{k_1\leq j<k}a_j=\underline a'_{k-k_1-1}.
\]
Hence $r'_{k-k_1}\in\calA_{s'}^{\mathrm{ret}}$, and $\Pi_sQ(\tau_s)$ is compatible with $s'$.
Moreover, $z_{s'-2}=a_{k_1}=\underline a_{k_1}<\underline a_{k_1-1}\leq a_0=z_{s-2}$.
On a triangle, $\calA_s^{\mathrm{ret}}=\emptyset$, so $\Pi_sQ(\tau_s)=\zeros$.
\end{proof}

\subsubsection{Cycle-draining phases and deterministic comparison}\label{sec:ctp-phases}\label{sec:ctp-refresh}

Next, we describe how source-draining stages form a cycle-draining phase.
For a cycle vertex $s$ and $\mu\in\mathbb Z_{\geq1}$, initialize both the genuine process $Q$ and an auxiliary process $\widetilde Q$ at singleton state $\mu\mathbf e_s$.
Both processes follow the $\CP$ policy and use the same arrivals, and we perform endpoint projections only for the auxiliary process.

Formally, for the $j$-th source-draining stage, we will use $S_j$ to denote its starting time, $q^{(j)}$ to denote its starting state, and $s_j$ to denote its source.
In particular, set $S_0=0$, $q^{(0)}=\mu\mathbf e_s$, and $s_0=s$.
At a source-draining stage starting at time $S_j$, let $X^{(j)}$ be the queue-length process under the $\CP$ policy started from $q^{(j)}$ and driven by the arrivals after $S_j$, and set
\[
 \tau_j:=\inf\{u\in\mathbb Z_{\geq1}:X^{(j)}_{s_j}(u)=0\}.
\]
If this set is empty, we set $\tau_j=\infty$.
For $0\leq u<\tau_j$, set $\widetilde Q(S_j+u)=X^{(j)}(u)$; if $\tau_j<\infty$, set
\[
 S_{j+1}=S_j+\tau_j,\qquad
 \widetilde Q(S_{j+1})=q^{(j+1)}=\Pi_{s_j}X^{(j)}(\tau_j).
\]
At a finite endpoint $\tau_j$, stop if $q^{(j+1)}=\zeros$; otherwise choose $s_{j+1}$ among its nonempty retained vertices to maximize $z_{i-2}$, equivalently taking the first nonempty retained vertex in the current path order, and start the next source-draining stage.
Proposition~\ref{lem:ctp-source} ensures that the projected state $q^{(j + 1)}$ is compatible with the new source $s_{j+1}$ and that $z_{s_{j+1}-2}<z_{s_j-2}$.
There are at most $n$ stages because the selected values $z_{s_j-2}$ strictly decrease.
Define
\begin{equation}\label{eq:ctp-singleton-time}
 \vartheta_s=\inf\{t\geq0:\widetilde Q(t)=\zeros\}\in\mathbb Z_{\geq0}\cup\{\infty\}
\end{equation}
as the first time the auxiliary process reaches the zero state.
We call the evolution up to $\vartheta_s$ a \emph{cycle-draining phase}.
The subscript $s$ denotes the initial source; the dependence on the singleton mass $\mu$ and the driving arrivals is suppressed.

The following proposition presents the drift and endpoint properties of cycle-draining phases.
In particular, at a finite endpoint, all queues in the auxiliary process are empty, while the genuine process may retain a residual bounded by the arrival-fluctuation error.

\begin{proposition}[Cycle-draining phase]\label{lem:ctp-singleton}\label{lem:ctp-disjoint}
For every cycle vertex $s$ and singleton initial state $Q(0)=\mu\mathbf e_s$ with $\mu\geq1$, for every finite integer $t\geq0$ and $0\leq v\leq\vartheta_s\wedge t$, pathwise,
\begin{equation}\label{eq:ctp-singleton-path}
 \norm{Q(v)}_1+2\epsilon v\leq \mu+O\!\left(n^3R(t)\right).
\end{equation}
On the event $\vartheta_s \leq t$, pathwise,
\begin{equation}\label{eq:ctp-singleton-end}
 \norm{Q(\vartheta_s)}_1=O\!\left(n^3R(t)\right).
\end{equation}
\end{proposition}

\begin{proof}
Fix a finite integer $t\geq0$ and the arrivals through $t$.
Since every stage considered below has its time interval lying in $[0,t]$, by~\eqref{eq:ctp-range-shift} its arrival-fluctuation error is at most $R(t)$.
Let $C$ be a constant for which~\eqref{eq:ctp-source-drift} and~\eqref{eq:ctp-projection} hold with error $Cn^2R(t)$.
At each finite endpoint $S_{j+1}=S_j+\tau_j$, let $d_j$ be the mass deleted by the projection, and let $\calD(v)$ be the cumulative deleted mass through time $v$:
\[
 d_j:=\|X^{(j)}(\tau_j)-q^{(j+1)}\|_1,
 \qquad \calD(v):=\sum_{j:S_{j+1}\leq v}d_j.
\]
Because a projection only deletes agents,
\[
 \norm{q^{(j+1)}}_1+d_j=\norm{X^{(j)}(\tau_j)}_1.
\]

Fix an integer $0\leq v\leq\vartheta_s\wedge t$, and let $J=\max\{j:S_j\leq v\}$.
Stages $0,\ldots,J-1$ are completed by time $v$, and $\widetilde Q(v)=X^{(J)}(v-S_J)$ with $v-S_J<\tau_J$; if $q^{(J)}=\zeros$, then $v=\vartheta_s$ and $X^{(J)}(0)=\zeros$.
Applying~\eqref{eq:ctp-source-drift} to each completed stage $j \in \{0, 1, \ldots, J - 1\}$ with $u=\tau_j$, and to the incomplete stage $J$ with $u=v-S_J$, gives
\[
 \begin{aligned}
 \norm{q^{(j+1)}}_1+d_j+2\epsilon\tau_j&\leq\norm{q^{(j)}}_1+Cn^2R(t)\qquad(j<J),\\
 \norm{\widetilde Q(v)}_1+2\epsilon(v-S_J)&\leq\norm{q^{(J)}}_1+Cn^2R(t).
 \end{aligned}
\]
Summing these $J+1\leq n+1$ inequalities, the masses $\lVert q^{(j)}\rVert_1$ telescope, the stage lengths add up to $v$, and the deleted masses add up to $\calD(v)$, so
\begin{equation}\label{eq:ctp-aux-drift}
 \norm{\widetilde Q(v)}_1+\calD(v)+2\epsilon v\leq\mu+C(n+1)n^2R(t).
\end{equation}
To compare the genuine process $Q$ and the auxiliary process $\widetilde Q$, start from $Q(0)=\widetilde Q(0)$ and $\calD(0)=0$.
Between projections, both processes follow the $\CP$ policy with the same arrivals, so Lemma~\ref{lem:prelim-coupling} prevents their distance from increasing.
At a projection, the state of the auxiliary process changes by $d_j$ in $\ell_1$ distance, so the triangle inequality increases the distance bound by at most $d_j$.
Induction over the projections therefore gives, for $0\leq v\leq\vartheta_s\wedge t$,
\begin{align}
 \norm{Q(v)}_1\leq \norm{\widetilde Q(v)}_1+\calD(v).
 \label{eqn:cou-gen-aux}
\end{align}
Combining \eqref{eqn:cou-gen-aux} with~\eqref{eq:ctp-aux-drift} proves~\eqref{eq:ctp-singleton-path}, including at projection times.

If $\vartheta_s\leq t$, then $\widetilde Q(\vartheta_s)=\zeros$ and every stage is completed within $[0,t]$, so~\eqref{eq:ctp-projection} gives $d_j\leq Cn^2R(t)$ for each of the at most $n$ stages, and hence $\calD(\vartheta_s)\leq Cn^3R(t)$.
The coupling bound \eqref{eqn:cou-gen-aux} now yields $\lVert Q(\vartheta_s)\rVert_1\leq\calD(\vartheta_s)$, proving~\eqref{eq:ctp-singleton-end}.
\end{proof}

We now apply cycle-draining phases to analyze a process starting from an arbitrary initial state.
The main difficulty of directly applying Proposition~\ref{lem:ctp-singleton} is that cycle-draining phases have stochastic horizons, and we get around it by analyzing a deterministic horizon.
For a state $q$, let $\mu(q):=\max_{i\in\calA_{\mathrm{cyc}}}q_i$ denote its largest cycle queue length.
The following proposition gives a drift bound for a process with an arbitrary initial state over a deterministic horizon.

\begin{proposition}\label{lem:ctp-core}\label{lem:ctp-cycle-only}
Suppose $Q(0)=q$ for an arbitrary $q\in\mathbb Z_{\geq0}^{\calA}$.
Then, for every finite integer $t\geq0$, pathwise,
\begin{equation*}
 \norm{Q(t)}_1\leq \norm{q}_1-\min\left\{\mu(q),2\epsilon t\right\}
 +O\!\left(n^4R(t)\right).
\end{equation*}
\end{proposition}

\begin{proof}
Let $\mu=\mu(q)$, and let $s$ be a cycle vertex attaining this maximum, i.e., $q_s = \mu$.
Let $Q^\circ$ be the process following the $\CP$ policy, with $Q^\circ(0)=\mu\mathbf e_s$ and the same arrivals.
The $\ell_1$-distance between $q$ and $\mu\mathbf e_s$ is $\lVert q\rVert_1-\mu$, so by Lemma~\ref{lem:prelim-coupling} and the identity $\mu-(\mu-2\epsilon t)^+=\min\{\mu,2\epsilon t\}$, it suffices to show that $Q^\circ$ satisfies
\begin{equation}\label{eq:ctp-sigma-time}
 \norm{Q^\circ(t)}_1\leq(\mu-2\epsilon t)^++O\!\left(n^4R(t)\right).
\end{equation}
Fix the arrivals through time $t$.
All non-cycle queues of $Q^\circ$ start empty, so \eqref{eq:ctp-branch-mass} gives the following pathwise bound for non-cycle queues:
\begin{align}
    \sup_{0\leq u\leq t}B(Q^\circ(u))=O(nR(t)).
    \label{eqn:non-cycle-bound-det}
\end{align}

We partition the trajectory of $Q^\circ$ into waiting steps and cycle-draining phases.
Whenever no phase is running and all cycle queues are empty, advance $Q^\circ$ by one arrival under the $\CP$ policy and check again.
Whenever no cycle-draining phase is running and a cycle queue is nonempty, we begin a new cycle-draining phase at the current time $S$: choose any cycle vertex $s'$ with the largest queue length, and write $x=\lVert Q^\circ(S)\rVert_1$ and $y=Q^\circ_{s'}(S)>0$.
Initialize the genuine singleton and auxiliary processes of Section~\ref{sec:ctp-phases} at $y\mathbf e_{s'}$, using the same arrivals after $S$ as $Q^\circ$, and let $\vartheta$ be the length~\eqref{eq:ctp-singleton-time} of this cycle-draining phase.
If $\vartheta<\infty$, the phase ends at $S+\vartheta$ and the same rule is applied to $Q^\circ(S+\vartheta)$; if $\vartheta=\infty$, no further phase begins.
When $\mu\geq1$, the first phase starts at time zero with $s'=s$ and $x=y=\mu$, and its genuine singleton process is $Q^\circ$ itself.

For a phase starting at $S\leq t$, apply Proposition~\ref{lem:ctp-singleton} to the arrivals after $S$, with local horizon $t-S$; by~\eqref{eq:ctp-range-shift}, its arrival-fluctuation error is controlled by $R(t)$.
Note that all guarantees are pathwise, so the random choice of $S$ causes no difficulty.
When $S+\vartheta\leq t$, write $x':=\lVert Q^\circ(S+\vartheta)\rVert_1$.
Since the initial distance between $Q^\circ(S)$ and $y\mathbf e_{s'}$ is $x-y$, Lemma~\ref{lem:prelim-coupling} gives
\begin{align}
 \sup_{S\leq u\leq(S+\vartheta)\wedge t}\norm{Q^\circ(u)}_1\leq x+O\!\left(n^3R(t)\right),
 x'\leq x-y+O\!\left(n^3R(t)\right)
                  \qquad(S+\vartheta\leq t).\label{eq:ctp-core-phase-end}
\end{align}
The cycle carries mass $x-B(Q^\circ(S))$ on at most $n$ vertices, so $y\geq(x-B(Q^\circ(S)))/n$; together with the non-cycle bound \eqref{eqn:non-cycle-bound-det}, \eqref{eq:ctp-core-phase-end} implies that every phase completed by time $t$ satisfies
\begin{equation}\label{eq:ctp-core-loss}
 x'\leq\left(1-\frac1n\right)x+Cn^3R(t)
\end{equation}
for a universal constant $C$.

Consider first the case $\mu\geq1$ and $\vartheta>t$ for the first phase.
Then~\eqref{eq:ctp-singleton-path} with $v=t$ gives $\lVert Q^\circ(t)\rVert_1\leq\mu-2\epsilon t+O(n^3R(t))$, which implies~\eqref{eq:ctp-sigma-time}.
In the remaining cases, either $\mu=0$ or the first phase ends by time $t$, and in the latter case~\eqref{eq:ctp-singleton-end} gives $\lVert Q^\circ(\vartheta)\rVert_1=O(n^3R(t))$.
We show that then every phase start $S\leq t$ other than time zero satisfies
\begin{equation}\label{eq:ctp-weighted-square}
 x=O\!\left(n^4R(t)\right).
\end{equation}
Suppose first that a phase starts at time $S$ immediately after a waiting step.
All cycle queues are empty at time $S-1$, so $\lVert Q^\circ(S-1)\rVert_1=B(Q^\circ(S-1))=O(nR(t))$, and the arrival at time $S$ is retained at some $i\in\calA_{\mathrm{cyc}}$.
The component has no slack, so $\lambda=Mz^*$ gives
\[
 \lambda_i=\sum_{e\in\calM:\,i\in e}z_e^*
 \leq\sum_{e\in\calM}z_e^*
 =\frac12\sum_{j\in\calA}\lambda_j\leq\frac12.
\]
The arrival range includes the interval $(S-1,S]$, hence $R(t)\geq|M_i(S)-M_i(S-1)|=1-\lambda_i\geq1/2$.
A single arrival increases the total mass by at most one, so $x\leq\lVert Q^\circ(S-1)\rVert_1+1\leq\lVert Q^\circ(S-1)\rVert_1+2R(t)=O(nR(t))$.
Now group the phases starting after time zero into maximal runs not separated by waiting steps, and let $x_0,\ldots,x_j$ be the starting masses in one run, so that $x_{h+1}$ is the end mass $x'$ in~\eqref{eq:ctp-core-phase-end} of the phase with starting mass $x_h$.
The first phase of a run follows either a waiting step or the phase starting at time zero, so $x_0=O(n^3R(t))$.
Iterating~\eqref{eq:ctp-core-loss} and summing the geometric series gives
\[
 x_j\leq\left(1-\frac1n\right)^jx_0+Cn^3R(t)\sum_{h=0}^{j-1}\left(1-\frac1n\right)^h
 =O\!\left(n^4R(t)\right),
\]
which proves~\eqref{eq:ctp-weighted-square}.
Finally, at time $t$ the process $Q^\circ$ is either in a waiting step, where all cycle queues are empty and $\lVert Q^\circ(t)\rVert_1=B(Q^\circ(t))=O(nR(t))$, or in a phase starting after time zero, where \eqref{eq:ctp-core-phase-end} and~\eqref{eq:ctp-weighted-square} give $\lVert Q^\circ(t)\rVert_1=O(n^4R(t))$.
This proves~\eqref{eq:ctp-sigma-time} in the remaining cases.
\end{proof}

\subsubsection{Quadratic drift and the proof of Theorem~\ref{thm:ctp-main}}\label{sec:ctp-completion}\label{sec:ctp-blocks}

We apply the deterministic comparison given by \Cref{lem:ctp-core} over a block that is long enough for the accumulated drift to dominate the arrival fluctuations.
The comparison holds on every finite prefix, so the block need not contain a complete cycle-draining phase.
The next lemma bounds the expected change in the squared total queue length over one block.

\begin{lemma}\label{lem:ctp-refresh}\label{eq:ctp-deadline}\label{eq:ctp-threshold}
Let $K$ be a sufficiently large universal constant, and set $\Delta:=\lceil Kn^9/\epsilon^2\rceil$ and $\eta:=2\epsilon\Delta$.
Suppose $Q(0)=q$ for an arbitrary $q\in\mathbb Z_{\geq0}^{\calA}$. Then,
\begin{equation}\label{eq:ctp-refresh-square}
 \Expect_q\norm{Q(\Delta)}_1^2-\norm{q}_1^2
 \leq-\frac{\eta}2\norm{q}_1+\eta B(q)+O(n\eta^2).
\end{equation}
Moreover, for every deterministic time $t_0\geq0$,
\begin{equation}\label{eq:ctp-refresh-conditional}
 \Expect\bigl[\norm{Q(t_0+\Delta)}_1^2\,\big|\,Q(t_0)\bigr]-\norm{Q(t_0)}_1^2
 \leq-\frac{\eta}2\norm{Q(t_0)}_1+\eta B(Q(t_0))+O(n\eta^2).
\end{equation}
\end{lemma}

The proof of Lemma~\ref{lem:ctp-refresh} is deferred to Appendix~\ref{sec:proof-refresh}.
We now complete the proof of Theorem~\ref{thm:ctp-main}.
In particular, Lemma~\ref{lem:ctp-refresh} and the bound~\eqref{eq:ctp-refreshed-branches} on the expected number of agents in non-cycle queues first give a bound on the time average of the expected total queue length.
We then establish a bound that holds uniformly over time and apply Lemma~\ref{lmm:opt-test} to obtain the regret bound.

\begin{proof}[Proof of Theorem~\ref{thm:ctp-main}]
When the active graph is a tree, Corollary~\ref{lem:ctp-tree-bound} gives the stronger queue-length bound $O(n^3/\epsilon)$.
It remains to establish the queue-length bound for odd-unicyclic active graphs.
Recall the parameters $\Delta,\eta$ from Lemma~\ref{lem:ctp-refresh}.
The bound~\eqref{eq:ctp-refreshed-branches} for empty-start non-cycle queues gives $\Expect_0B(Q(u))=O(n^3/\epsilon)=O(\eta)$ at every deterministic time $u$, so taking expectations in~\eqref{eq:ctp-refresh-conditional} with $t_0=u$ gives
\begin{align}
 \Expect_0\norm{Q(u+\Delta)}_1^2-\Expect_0\norm{Q(u)}_1^2
 \leq-\frac{\eta}2\Expect_0\norm{Q(u)}_1+O(n\eta^2).
 \label{eqn:drift-ine-thm-2}
\end{align}
For any integer $t\geq1$, it follows that
 \begin{align}
 \sum_{u=0}^{t-1}\left(\Expect_0\norm{Q(u+\Delta)}_1^2-\Expect_0\norm{Q(u)}_1^2\right)
 &=\sum_{u=\Delta}^{t+\Delta-1}\Expect_0\norm{Q(u)}_1^2
       -\sum_{u=0}^{t-1}\Expect_0\norm{Q(u)}_1^2 \notag\\
 &=\sum_{u=t}^{t+\Delta-1}\Expect_0\norm{Q(u)}_1^2
       -\sum_{u=0}^{\Delta-1}\Expect_0\norm{Q(u)}_1^2.
       \label{eqn:iden-tele-sum-drift}
 \end{align}
Summing the drift inequality \eqref{eqn:drift-ine-thm-2} over $u=0,\ldots,t-1$, substituting the above identity \eqref{eqn:iden-tele-sum-drift}, and rearranging therefore yields
\begin{equation*}
 \begin{aligned}
 \frac1t\sum_{u=0}^{t-1}\Expect_0\norm{Q(u)}_1
 &\leq O(n\eta)+\frac{2}{\eta t} \left( \sum_{u=0}^{\Delta-1}\Expect_0\norm{Q(u)}_1^2- \sum_{u=t}^{t+\Delta-1}\Expect_0\norm{Q(u)}_1^2 \right)\\
 &\leq O(n\eta)+\frac{2}{\eta t}\sum_{u=0}^{\Delta-1}\Expect_0\norm{Q(u)}_1^2
 \leq O(n \eta) + O \left( \frac{\Delta^3}{\eta t} \right),
 \end{aligned}
\end{equation*}
where the last inequality holds since $\lVert Q(u)\rVert_1\leq u$.
Since $\epsilon\leq1$, we have $Kn^9/\epsilon^2\leq\Delta\leq2Kn^9/\epsilon^2$, and hence $\eta = \Theta(n^9 / \epsilon)$.
For fixed $\Delta$ and $\eta$, the term $O(\Delta^3/(\eta t))$ vanishes as $t\to\infty$, giving a limiting upper average of $O(n\eta)=O(n^{10}/\epsilon)$ for the expected total queue length.

To pass from the time-average bound to a bound that holds uniformly over time, we first construct an invariant distribution with bounded expected total queue length.
Let $P$ be the transition kernel and $\nu_t=t^{-1}\sum_{u<t}\mathcal L_0(Q(u))$ for $t\geq1$.
Write $\delta_0$ for the probability distribution concentrated at the empty state, so $\mathcal L_0(Q(0))=\delta_0$.
The average bound and finite sublevel sets of $\lVert \cdot\rVert_1$ on the state space make $(\nu_t)$ tight, so a subsequence converges in total variation to a probability distribution $\pi$.
Since $\nu_tP-\nu_t=(\mathcal L_0(Q(t))-\delta_0)/t$, this limit is invariant.
Every $\nu_t$ is supported on the class reachable from the empty state, so its total-variation limit $\pi$ is supported on that class.
Fatou's lemma and the average bound give
\begin{equation*}
 \Expect_\pi \norm{Q}_1=O(n^{10}/\epsilon).
\end{equation*}
Stationary coupling~\eqref{eq:prelim-stationary-empty} now yields
\begin{equation*}
 \sup_{t\geq0}\Expect_0\norm{Q(t)}_1
 \leq2\Expect_\pi \norm{Q}_1=O(n^{10}/\epsilon).
\end{equation*}
The empty-reachable class is irreducible: successive arrivals at active neighbors of occupied queues reduce total mass by one per step and reach zero with positive probability.
The invariant probability distribution therefore implies positive recurrence~\cite{gallager2011countable}; aperiodicity is not required.

The policy uses only active matches and discards only at under-demanded roots, so Lemma~\ref{lmm:opt-test} and the uniform queue bounds give $\regret(\mathbf{GP},t)=O(r_{\max}n^{10}/\epsilon)$.
\end{proof}

\bibliographystyle{alpha}
\bibliography{references}

\clearpage
\appendix

\section{Proof of Lemma~\ref{lmm:opt-test}}\label{sec:ctp-regret-proof}

Let $p\in\reals_{\geq0}^n$ be an optimal dual vector for the static-planning problem~\eqref{eq:static-planning}:
\begin{equation*}
 \min\{\lambda^Tp:M^Tp\geq r,\ p\geq0\}.
\end{equation*}
Complementary slackness gives $p_i+p_j=r_{\{i,j\}}$ on active edges and $p_i=0$ for $i\in\calA_+$.
For every $i\in\calA_0$, we have $s_i^*=0$, so the $i$-th constraint in~\eqref{eq:static-planning} gives
\[
 0<\lambda_i=\sum_{m\in\calM}M_{im}z_m^*.
\]
Since all summands are nonnegative, some match $m=\{i,j\}$ incident to $i$ satisfies $z_m^*>0$, and hence is active.
Complementary slackness and $p_j\geq0$ then imply
\[
 0\leq p_i\leq p_i+p_j=r_m\leq r_{\max}.
\]

To bound the hindsight reward, fix a realization of $A(u)$ and any matching vector $y$ feasible for~\eqref{eq:def-hin-opt}.
Then $y\geq0$ and $My\leq A(u)$.
By the dual constraints $M^Tp\geq r$ and $p\geq0$, we have
\[
 r^Ty\leq(M^Tp)^Ty=p^TMy\leq p^TA(u).
\]
Maximizing over all feasible hindsight matchings and then taking expectations yields the weak-duality bound
\[
 \calR^*(u)\leq\Expect[p^TA(u)].
\]
Since the policy uses only active matches, complementary slackness gives $r^TD(u)=p^TMD(u)$ and hence $\calR^\Pi(u)=\Expect[p^TMD(u)]$.
Discards occur only at types $i\in\calA_+$, where $p_i=0$, so $p^TF(u)=0$.
Using $Q(u)=A(u)-MD(u)-F(u)$ therefore gives
\[
 \begin{aligned}
 \calR^*(u)-\calR^\Pi(u)
 &\leq\Expect[p^T(A(u)-MD(u))]\\
 &=\Expect[p^TQ(u)]
 \leq r_{\max}\sum_{i\in\calA_0}\Expect Q_i(u)
 \leq r_{\max}\bar B.
 \end{aligned}
\]
Taking the supremum over $0\leq u\leq t$ proves the lemma.

\section{Deferred proofs in Section~\ref{sec:prelim-tools}}\label{sec:prelim-proofs}

\subsection{Proof of Lemma~\ref{lem:prelim-coupling}}

To prove \eqref{eq:prelim-coupling}, it suffices to consider one period of the two coupled processes $Q$ and $Q'$.
Suppose that the arrival has type $i$.
Under $\mathbf{GP}$, each process either matches the arrival to a nonempty neighbor of $i$, retains it at $i$, or discards it, and discarding occurs only when $i$ is a truncated root.
Whether an unmatched arrival is retained or discarded depends only on $i$, so if neither process matches the arrival, both take the same action.
If both processes take the same action, their states change by the same vector and hence the $\ell_1$-distance between them is unchanged.
Otherwise, after interchanging $Q$ and $Q'$ if necessary, suppose $Q$ matches the arrival to a queue $j$, while $Q'$ either matches it to a queue of lower priority for $i$ or does not match it.
Both processes use the same priority order and $Q'$ did not choose $j$, so $q_j>0=q'_j$, and the match in $Q$ reduces the distance by exactly one.
The update of $Q'$ changes at most one coordinate by one, so it increases the distance by at most one.
Hence the distance does not increase, and induction over periods proves~\eqref{eq:prelim-coupling}.

For~\eqref{eq:prelim-stationary-empty}, take $Q(0)=0$ and sample $Q'(0)\sim\pi$ independently of the future arrivals.
Couple $Q$ and $Q'$ using the same arrivals.
By the triangle inequality,
\[
 \norm{Q(t)}_1\leq \norm{Q'(t)}_1+\|Q(t)-Q'(t)\|_1
 \leq \norm{Q'(t)}_1+\norm{Q'(0)}_1,
\]
where the second inequality follows from~\eqref{eq:prelim-coupling} and $Q(0)=0$.
Taking expectations and using stationarity of $Q'$ gives the uniform bound~\eqref{eq:prelim-stationary-empty}.

\subsection{Proof of Lemma~\ref{lem:prelim-fluctuations}}

For each $j\in\calA$, the bound $R_j(t)\leq2\max_{0\leq s\leq t}|M_j(s)|$, Doob's $L^2$ inequality, and independence of the arrivals give
\[
 \Expect[R_j(t)^2]
 \leq16\lambda_j t.
\]
Cauchy--Schwarz and $\sum_{j\in\calA}\lambda_j\leq1$ now give~\eqref{eq:prelim-range-second-moment}:
\[
 \Expect[R(t)^2]
 \leq n\sum_{j\in\calA}\Expect[R_j(t)^2]
 \leq16nt.
\]

Cauchy--Schwarz gives $\Expect R(t)\leq\sqrt{\Expect[R(t)^2]}=O(\sqrt{nt})$, proving~\eqref{eq:prelim-range-first-moment}.
Monotonicity and the pathwise increment bound follow directly from the definition of $R_j(t)$.

\subsection{Proof of Lemma~\ref{lem:prelim-reflected-walk}}

Fix $\delta>0$, a vertex $j$, and a sign value $\sigma\in\{-1,1\}$.
Write the increment at time $s$ as
\[
 x_s:=\sigma\bigl[M_j(s)-M_j(s-1)\bigr]-\delta.
\]
The recursion~\eqref{eq:prelim-reflected-walk} adds $x_t$ to the current value and replaces the result by zero whenever it is negative.
Starting from $H_j^\sigma(0)=0$, we obtain
\[
 H_j^\sigma(t)=\max_{0\leq\ell\leq t}\sum_{s=\ell+1}^{t}x_s.
\]
Indeed, the identity holds at $t=0$.
If it holds at $t-1$, then
\[
\begin{aligned}
 H_j^\sigma(t)
 =\max\left\{0,\max_{0\leq\ell\leq t-1}
 \left(\sum_{s=\ell+1}^{t-1}x_s+x_t\right)\right\}
 =\max_{0\leq\ell\leq t}\sum_{s=\ell+1}^{t}x_s,
\end{aligned}
\]
where $\ell=t$ corresponds to the empty sum, which equals zero.
For each $\ell$, the increments of $M_j$ telescope:
\[
 \sum_{s=\ell+1}^{t}x_s
 =\sigma\bigl[M_j(t)-M_j(\ell)\bigr]-\delta(t-\ell).
\]
Substituting this telescoping identity into the maximum representation gives~\eqref{eq:prelim-reflected-maximum}.
Dropping the nonpositive drift term in~\eqref{eq:prelim-reflected-maximum} bounds $H_j^\sigma(\ell)$ by $R_j(t)$ for every $0\leq\ell\leq t$.

For the first-moment bound, we use the following claim.
\begin{claim}\label{cla:prelim-walk-bound}
Let $X_1,X_2,\ldots$ be i.i.d. with $\Expect X_h=0$ and $|X_h|\leq1$.
Set $S_0=0$ and $S_a:=\sum_{h=1}^a X_h$ for $a\geq1$.
Then, for every $\delta>0$,
\begin{equation}\label{eq:prelim-walk-bound}
 \Expect\sup_{a\geq0}\{S_a-\delta a\}=O(1/\delta).
\end{equation}
\end{claim}

To apply Claim~\ref{cla:prelim-walk-bound}, let
\[
 X_h:=\sigma\bigl[M_j(h)-M_j(h-1)\bigr]
 =\sigma(\1{\theta_h=j}-\lambda_j).
\]
Because the arrival types $\theta_h$ are i.i.d.\ with $\Pr(\theta_h=j)=\lambda_j$, these increments are i.i.d., have mean zero, and lie in $[-1,1]$.
For each fixed $t$, reversing $X_1,\ldots,X_t$ preserves their joint law, so~\eqref{eq:prelim-reflected-maximum} gives
\[
 H_j^\sigma(t)\overset{d}{=}\max_{0\leq a\leq t}\{S_a-\delta a\}.
\]
This maximum is bounded by the supremum in~\eqref{eq:prelim-walk-bound}, so Claim~\ref{cla:prelim-walk-bound} gives $\sup_{t\geq0}\Expect H_j^\sigma(t)=O(1/\delta)$, proving the uniform first-moment bound.

\begin{proof}[Proof of Claim~\ref{cla:prelim-walk-bound}]
Kingman's bound, expressed for the maximum of a random walk with negative drift~\cite{kingman1962inequalities} (see also~\cite[Sections~2.1--2.2, Eq.~(2.6)]{chen2020algorithms}), gives
\[
 \Expect\sup_{a\geq0}\{S_a-\delta a\}
 \leq\frac{\Var(X_1)}{2\delta}
 \leq\frac{1}{2\delta},
\]
where the last inequality uses $\Var(X_1)\leq1$ because $|X_1|\leq1$.
\end{proof}

\section{Deferred proofs in Section~\ref{sec:ctp-proof}}\label{sec:cycle-proofs}

\subsection{Proof of Lemma~\ref{lem:ctp-branches}}\label{sec:proof-branches}

The pathwise bound in Lemma~\ref{lem:prelim-reflected-walk} and Proposition~\ref{lem:ctp-timestamps} give, for every $i\in\calA$,
\begin{equation}\label{eq:ctp-coordinate-counting}
 \sup_{0\leq t\leq U}Y_i(Q(t))^+
 \leq\sum_{j\in\calT(i)}R_j(U).
\end{equation}
For each $i\in\calA_{\mathrm{cyc}}$, grouping the sum in~\eqref{eq:ctp-workload} by child subtrees gives $|Y_i|\leq Y_i^++\sum_{j\in\calC(i)}Y_j^+$.
For a fixed cycle vertex $i$, the subtrees $\calT(j)$ of its children partition $\calT_i\setminus\{i\}$.
Applying~\eqref{eq:ctp-coordinate-counting} to $Y_i^+$ and to each $Y_j^+$ therefore counts every vertex of $\calT_i$ at most twice in the bound for $|Y_i|$.
The attached trees $\calT_i$, indexed by $i\in\calA_{\mathrm{cyc}}$, partition $\calA$, so summing gives
\[
 \sum_{i\in\calA_{\mathrm{cyc}}}\sup_{0\leq t\leq U}|Y_i(t)|
 \leq2\sum_{i\in\calA_{\mathrm{cyc}}}\sum_{j\in\calT_i}R_j(U)
 =2\sum_{j\in\calA}R_j(U)
 =2R(U),
\]
which proves~\eqref{eq:ctp-branch-imbalance}.

For the total length of the non-cycle queues, apply the subtree inventory identity~\eqref{eq:ctp-subtree-inventory} to each attached tree and sum over $i\in\calA_{\mathrm{cyc}}$:
\[
 \begin{aligned}
 B(Q(t))
 =\sum_{i\in\calA_{\mathrm{cyc}}}Y_i(t)
   +2\sum_{i\in\calA\setminus\calA_{\mathrm{cyc}}}Y_i(t)
 \leq2\sum_{i\in\calA}Y_i(t)^+.
 \end{aligned}
\]
Each vertex $j$ belongs to $\calT(i)$ only when $i$ is an ancestor of $j$ in its attached tree, including $i=j$, so it occurs in at most $n$ such subtrees.
Taking the supremum and using~\eqref{eq:ctp-coordinate-counting} thus gives
\[
 \sup_{0\leq t\leq U}B(Q(t))
 \leq2\sum_{i\in\calA}\sum_{j\in\calT(i)}R_j(U)
 \leq2nR(U),
\]
proving~\eqref{eq:ctp-branch-mass}.

For the empty-start process, Proposition~\ref{lem:ctp-timestamps} and Lemma~\ref{lem:prelim-reflected-walk} give, for every $i\in\calA$,
\begin{equation}\label{eq:ctp-initial-branches}
 \sup_{t\geq0}\Expect_0\!\left[Y_i(Q(t))^+\right]
 =O(n^2/\epsilon).
\end{equation}
Summing~\eqref{eq:ctp-initial-branches} over $i\in\calA$ and using $B\leq2\sum_{i\in\calA}Y_i^+$, as established above, proves~\eqref{eq:ctp-refreshed-branches}.

\subsection{Proof of Lemma~\ref{fact:ctp-queue-profile}}\label{sec:ctp-profile-proof}

Fix a source-draining stage with initial state $q=Q(0)$ compatible with source $s$ and a finite $U\leq\tau_s$. Recall the fluid approximation $\bar q(t)$ defined in~\eqref{eq:ctp-comparison-profile}, the boundary flows~\eqref{eq:ctp-cut-flow}--\eqref{eq:ctp-path-flows}, and the running minima~\eqref{eq:ctp-running-minima}.
We first state two queue-length identities and a bound on the prefix errors, which we use to prove Lemma~\ref{fact:ctp-queue-profile}.
The proofs of these three claims are given at the end of this subsection.

\begin{claim}\label{cla:ctp-vertex-identity}
For every $i\in\calA_{\mathrm{cyc}}$ and $0\leq t\leq U$,
\begin{equation}\label{eq:ctp-core-balance}
 \begin{aligned}
 Q_i(t)-q_i
 &=(z_{i-1}+z_i)t+\sum_{j\in\calT_i}(-1)^{d(i,j)}M_j(t)+Y_i(t)\\
 &\quad-\#\{\text{cycle matches involving }i\text{ up to }t\}.
 \end{aligned}
\end{equation}
\end{claim}

As explained in Section~\ref{sec:ctp-overview}, the edge $\{s-2,s-1\}$ is unused while $Q_s>0$.
The resulting path, with $r_{m+1}=s+1$, is
\[
 r_1-u_1-\cdots-r_m-u_m-r_{m+1}-s-(s-1).
\]
For a triangle, the pairs are absent and only the last three vertices remain. If $\tau_s<\infty$, both source neighbors are still empty at $\tau_s$, since the arrival that empties the source is matched immediately.

A prefix is an initial segment of the opened path, starting at $r_1$.
In particular, an even prefix contains the first $k$ pairs; the next odd prefix also includes their neighbor $r_{k+1}$.
Define the error for the first $k$ pairs as
\begin{equation}\label{eq:ctp-pair-inventory}
 E_k(t):=\sum_{j=1}^k\bigl(Q_{r_j}(t)-Q_{u_j}(t)-\bar q_{r_j}(t)\bigr)
 \quad(1\leq k\leq m),
\end{equation}
with $E_0(t):=0$.
Thus $E_k(t)$ is the difference between the alternating queue sum over the first $k$ pairs and $\sum_{j=1}^k\bar q_{r_j}(t)$, and $E_k(0)=0$.

\begin{claim}\label{cla:ctp-prefix-identity}
For every $1\leq k\leq m$ and $0\leq t_0\leq t\leq U$,
\begin{equation}\label{eq:ctp-pair-balance}
 \begin{aligned}
 E_k(t)-E_k(t_0)&=-(a_k-\underline a_k)(t-t_0)\\
 &\quad+\#\{\text{matches across }\{u_k,r_{k+1}\}\text{ in }(t_0,t]\}
 +O\!\left(R(U)\right).
 \end{aligned}
\end{equation}
\end{claim}

In the following claim, we bound $\sup_{0\leq t\leq U}|E_k(t)|$ by the arrival-fluctuation error.

\begin{claim}\label{cla:cumulative-profile-estimate}
For every finite $U\leq\tau_s$ and every $0\leq k\leq m$,
\begin{equation}\label{eq:ctp-record-error}
 \sup_{0\leq t\leq U}|E_k(t)|
 =O\!\left(kR(U)\right).
\end{equation}
\end{claim}

Now, we are ready to prove Lemma~\ref{fact:ctp-queue-profile}.

\begin{proof}[Proof of Lemma~\ref{fact:ctp-queue-profile}]
We first show that, for every finite $U\leq\tau_s$, the source error satisfies
\begin{equation}\label{cla:ctp-source-error}
 \sup_{0\leq t\leq U}\bigl|Q_s(t)-\bar q_s(t)-E_m(t)\bigr|
 =O\!\left(R(U)\right).
\end{equation}
Fix $0\leq t\leq U$ and sum equations~\eqref{eq:ctp-core-balance} over the full cut path, giving positive signs to $r_j$, $1\leq j\leq m+1$, and $s-1$, and negative signs to $u_j$, $1\leq j\leq m$, and $s$.
Every used cycle match joins vertices with opposite signs, so its match count cancels; the removed edge $\{s-2,s-1\}$ carries no matches.
The sum of the terms involving $M_j(t)$ and $Y_i(t)$ is bounded in absolute value, uniformly for $0\leq t\leq U$, by
\[
 \begin{aligned}
 \sum_{i\in\calA_{\mathrm{cyc}}}
 \left[\sum_{j\in\calT_i}|M_j(t)|+|Y_i(t)|\right]
 \leq R(U)+\sum_{i\in\calA_{\mathrm{cyc}}}|Y_i(t)|
 =O\left(R(U)\right),
 \end{aligned}
\]
where the inequality uses disjointness of the attached trees and $|M_j(t)|\leq R_j(U)$, and the final bound follows from~\eqref{eq:ctp-branch-imbalance}.
The signed sum of~\eqref{eq:ctp-core-balance} therefore gives
\[
 \begin{aligned}
 &\sum_{j=1}^m
 \bigl[(Q_{r_j}(t)-q_{r_j})-(Q_{u_j}(t)-q_{u_j})\bigr]
 +(Q_{s+1}(t)-q_{s+1})-(Q_s(t)-q_s)
 +(Q_{s-1}(t)-q_{s-1})\\
 &=\Bigl[\sum_{j=1}^m(a_{j-1}-a_j)+(z_s+z_{s+1})
 -(z_{s-1}+z_s)+(z_{s-2}+z_{s-1})\Bigr]t
 +O\!\left(R(U)\right)\\
 &=2a_0t+O\!\left(R(U)\right),
 \end{aligned}
\]
where the last equality uses $\sum_{j=1}^m(a_{j-1}-a_j)=a_0-a_m$, $z_{s+1}=a_m$, and $z_{s-2}=a_0$.
Since both source neighbors have zero queues and zero initial inventories, and compatibility gives $q_{u_j}=0$ for $1\leq j\leq m$, this simplifies to
\[
 \begin{aligned}
 \sum_{j=1}^m\bigl(Q_{r_j}(t)-Q_{u_j}(t)\bigr)-Q_s(t)
 =\sum_{j=1}^m q_{r_j}-q_s+2a_0t
 +O\!\left(R(U)\right).
 \end{aligned}
\]
For the fluid approximation~\eqref{eq:ctp-comparison-profile}, summing its $r_j$ coordinates and subtracting its source coordinate gives
\[
 \begin{aligned}
 \sum_{j=1}^m\bar q_{r_j}(t)-\bar q_s(t)
 =\sum_{j=1}^m q_{r_j}-q_s
   +\bigl[(a_0-\underline a_m)+(a_0+\underline a_m)\bigr]t
 =\sum_{j=1}^m q_{r_j}-q_s+2a_0t.
 \end{aligned}
\]
Subtracting this identity from the preceding displayed equation cancels the initial terms and $2a_0t$; by~\eqref{eq:ctp-pair-inventory}, the remaining left-hand side is $E_m(t)-(Q_s(t)-\bar q_s(t))$.
Rearranging therefore gives
\begin{align*}
 Q_s(t)-\bar q_s(t)
 =E_m(t)+O\!\left(R(U)\right),
\end{align*}
which proves~\eqref{cla:ctp-source-error} because the bound holds uniformly for $0\leq t\leq U$.

We now bound the errors in each pair and combine them with the source estimate.
Consecutive prefix errors give $Q_{r_k}-Q_{u_k}-\bar q_{r_k}=E_k-E_{k-1}$; since adjacent queues cannot both be positive and $\bar q_{r_k}\geq0$, this implies
\begin{equation}\label{eq:ctp-physical-error}
 |Q_{r_k}(t)-\bar q_{r_k}(t)|+Q_{u_k}(t)
 =|E_k(t)-E_{k-1}(t)|.
\end{equation}
For pair $k$, equation~\eqref{eq:ctp-physical-error}, the triangle inequality, and Claim~\ref{cla:cumulative-profile-estimate} give, for every $t\leq U$,
\[
 \begin{aligned}
 |Q_{r_k}(t)-\bar q_{r_k}(t)|+Q_{u_k}(t)
 \leq |E_k(t)|+|E_{k-1}(t)|
 =O\!\left(kR(U)\right).
 \end{aligned}
\]
For the source, Claim~\ref{cla:cumulative-profile-estimate} and~\eqref{cla:ctp-source-error} give
\[
 |Q_s(t)-\bar q_s(t)|=O\!\left((m+1)R(U)\right).
\]
It remains to sum these estimates and include the non-cycle queues.
For every $j\in\calA\setminus\calA_{\mathrm{cyc}}$, definition~\eqref{eq:ctp-comparison-profile} sets $\bar q_j(t)=0$, so these vertices contribute exactly $B(Q(t))$ to $\|Q(t)-\bar q(t)\|_1$.
The two source neighbors contribute zero because both their actual queues and their coordinates in $\bar q(t)$ vanish.
The remaining vertices are the $m$ pairs and the source, so summing their error bounds gives
\[
 \begin{aligned}
 \sup_{0\leq t\leq U}\|Q(t)-\bar q(t)\|_1
 \leq O\!\left((m+1)^2R(U)\right)+\sup_{0\leq t\leq U}B(Q(t))
 \leq O\!\left(n^2R(U)\right) \,,
 \end{aligned}
\]
where the last inequality uses the bound on the total length of the non-cycle queues in~\eqref{eq:ctp-branch-mass} and $m\leq n$.
This includes a triangle, with $m=0$ and $E_0=0$.
This proves Lemma~\ref{fact:ctp-queue-profile}.
\end{proof}

We now prove Claims~\ref{cla:ctp-vertex-identity}--\ref{cla:cumulative-profile-estimate} in order.

\begin{proof}[Proof of Claim~\ref{cla:ctp-vertex-identity}]
Fix $i\in\calA_{\mathrm{cyc}}$ and $0\leq t\leq U$, and let $D_i^{\mathrm{child}}(t)$ count the matches between $i$ and its children in the attached tree during $(0,t]$.
Recall that
\[
 Y_i(t)=\sum_{j\in\calT_i\setminus\{i\}}(-1)^{d(i,j)+1}Q_j(t).
\]
Thus agents at odd distance from $i$ have weight $+1$, those at even distance have weight $-1$, and agents at $i$ are excluded.
For this accounting, first add the arriving agent and then remove the two agents used in a match, if one occurs.
\begin{itemize}[nosep,topsep=3pt]
 \item An arrival at $j\in\calT_i\setminus\{i\}$ adds weight $(-1)^{d(i,j)+1}$, including when the arrival is immediately matched.
 \item A match with both endpoints in $\calT_i\setminus\{i\}$ removes one agent of weight $+1$ and one of weight $-1$, so the matching operation contributes zero.
 \item A match between $i$ and a child contributes $-1$ in the removal step, because only the positively weighted child is counted.
\end{itemize}
Since all non-cycle queues start empty, $Y_i(0)=0$, and summing these contributions gives
\[
 \begin{aligned}
 Y_i(t)
 \mathrel{=}\sum_{j\in\calT_i\setminus\{i\}}(-1)^{d(i,j)+1}A_j(t)-D_i^{\mathrm{child}}(t)
 =-\sum_{j\in\calT_i\setminus\{i\}}(-1)^{d(i,j)}A_j(t)
 -D_i^{\mathrm{child}}(t).
 \end{aligned}
\]

Arrivals and these matches with children together contribute $A_i(t)-D_i^{\mathrm{child}}(t)$ to the change in the queue at $i$.
By the cycle case of~\eqref{eq:ctp-drift-parameters},
\[
 \sum_{j\in\calT_i}(-1)^{d(i,j)}\lambda_j=z_{i-1}+z_i.
\]
Substituting these identities and using $M_j(t)=A_j(t)-\lambda_jt$ gives
\begin{equation*}
 A_i(t)-D_i^{\mathrm{child}}(t)
 =(z_{i-1}+z_i)t+\sum_{j\in\calT_i}(-1)^{d(i,j)}M_j(t)+Y_i(t).
\end{equation*}
Subtracting the cycle-match counts from $A_i(t)-D_i^{\mathrm{child}}(t)$ gives \eqref{eq:ctp-core-balance}, proving the claim.
\end{proof}

\begin{proof}[Proof of Claim~\ref{cla:ctp-prefix-identity}]
Fix $1\leq k\leq m$ and $0\leq t_0\leq t\leq U$.
For each $1\leq j\leq k$, subtract equation~\eqref{eq:ctp-core-balance} for $u_j$ from the same equation for $r_j$.
Since $u_j=r_j-1$, the coefficient of $t$ in this difference is
\[
 \begin{aligned}
 (z_{r_j-1}+z_{r_j})-(z_{u_j-1}+z_{u_j})
 =z_{r_j}-z_{r_j-2}
 =a_{j-1}-a_j,
 \end{aligned}
\]
where the last equality uses~\eqref{eq:ctp-cut-flow} and~\eqref{eq:ctp-path-flows}.
Summing these deterministic input rates over the first $k$ pairs telescopes to $\sum_{j=1}^k(a_{j-1}-a_j)=a_0-a_k$.

Disjointness of the attached trees and~\eqref{eq:ctp-branch-imbalance} bound the total fluctuation contribution by $O(R(U))$.
For the match terms, counts from edges with both endpoints in the prefix appear with coefficients $-1$ and $+1$ in the signed sum and therefore cancel.
The left cut edge is unused, while each match across the right boundary $\{u_k,r_{k+1}\}$ has coefficient $+1$ in the signed sum because the inside endpoint $u_k$ has negative sign.
Compatibility of the initial state gives $q_{u_j}=0$ for each $j$, so summing all terms of~\eqref{eq:ctp-core-balance} over the prefix with alternating signs gives
\begin{equation*}
 \begin{aligned}
 \sum_{j=1}^k\bigl(Q_{r_j}(t)-Q_{u_j}(t)\bigr)
 &=\sum_{j=1}^k q_{r_j}+(a_0-a_k)t+O\!\left(R(U)\right)\\
 &\quad+\#\{\text{matches across }\{u_k,r_{k+1}\}\text{ by time }t\}.
 \end{aligned}
\end{equation*}
By~\eqref{eq:ctp-comparison-profile}, every $\bar q_{u_j}(t)$ is zero, so the alternating sum of the coordinates of $\bar q(t)$ over the same prefix is $\sum_{j=1}^k\bar q_{r_j}(t)$.
Substituting the formula for $\bar q_{r_j}(t)$ and using $\underline a_0=a_0$ gives
\[
 \begin{aligned}
 \sum_{j=1}^k\bar q_{r_j}(t)
 =\sum_{j=1}^k q_{r_j}
   +t\sum_{j=1}^k(\underline a_{j-1}-\underline a_j)
 =\sum_{j=1}^k q_{r_j}+(a_0-\underline a_k)t.
 \end{aligned}
\]
Subtracting the preceding deterministic identity from the queue-length identity cancels the initial terms $\sum_{j=1}^k q_{r_j}$, while the coefficient of $t$ becomes
\[
 (a_0-a_k)-(a_0-\underline a_k)=-(a_k-\underline a_k).
\]
By~\eqref{eq:ctp-pair-inventory}, it follows that
\[
 \begin{aligned}
 E_k(t)
 & =-(a_k-\underline a_k)t\\
 &\quad+\#\{\text{matches across }\{u_k,r_{k+1}\}\text{ by time }t\}
 +O\!\left(R(U)\right).
 \end{aligned}
\]
Evaluate this identity at $t_0$ and subtract it from the identity at $t$.
The difference of the match counts is exactly the number of matches across $\{u_k,r_{k+1}\}$ during $(t_0,t]$.
This proves \eqref{eq:ctp-pair-balance}.
\end{proof}

\begin{proof}[Proof of Claim~\ref{cla:cumulative-profile-estimate}]

Write $W_k:=\sup_{0\leq t\leq U}|E_k(t)|$, so $W_0=0$.
For each $1\leq k\leq m$, we will prove the following bounds uniformly for $0\leq t\leq U$:
\[
 \begin{aligned}
 E_k(t)\leq W_{k-1}+O\!\left(R(U)\right) \quad \text{and} \quad
 E_k(t)\geq -W_{k-1}-O\!\left(R(U)\right).
 \end{aligned}
\]
Together, the upper and lower bounds give $|E_k(t)|\leq W_{k-1}+O\!\left(R(U)\right)$ uniformly for $0\leq t\leq U$.
Taking the supremum over $t\leq U$ yields $W_k\leq W_{k-1}+O\!\left(R(U)\right)$; iterating from $W_0=0$ then gives $W_k=O\!\left(kR(U)\right)$, proving~\eqref{eq:ctp-record-error}.

\paragraph{Upper bound on $E_k(t)$.}
Use the prefix ending at $u_k$, whose right boundary is $\{u_k,r_{k+1}\}$.
Fix $t\leq U$, and let $t_0$ be the last empty time of $r_k$ in $[0,t]$, or zero if it has never been empty.
If $Q_{r_k}(t_0)=0$, definition~\eqref{eq:ctp-pair-inventory} and $\bar q_{r_k}(t_0)\geq0$ give
\[
 E_k(t_0)-E_{k-1}(t_0)=-Q_{u_k}(t_0)-\bar q_{r_k}(t_0)\leq0.
\]
Otherwise, $r_k$ has never been empty and $t_0=0$; since $Q(0)=q=\bar q(0)$, the two prefix errors are $E_k(0)=E_{k-1}(0)=0$.
Thus $E_k(t_0)\leq E_{k-1}(t_0)$ in either case.
On $(t_0,t]$, the queue at $r_k$ stays positive, so its neighbor $u_k$ is empty.
Arrivals at $u_k$ choose an attached child or $r_k$ before $r_{k+1}$, and hence arrivals at $r_{k+1}$ cannot match the empty queue $u_k$.
This rules out matches across $\{u_k,r_{k+1}\}$ in $(t_0,t]$.
Since $\underline a_k\leq a_k$, the interval identity~\eqref{eq:ctp-pair-balance} and the inequality $E_k(t_0)\leq E_{k-1}(t_0)$ established above give, also when $t_0=t$,
\begin{align*}
 E_k(t)\leq E_k(t_0)+O\!\left(R(U)\right)
 &\leq E_{k-1}(t_0)+O\!\left(R(U)\right)
 \leq W_{k-1}+O\!\left(R(U)\right).
\end{align*}

\paragraph{Lower bound on $E_k(t)$.}
If $r_k\in\calA_s^{\mathrm{ret}}$, then $\underline a_k=a_k$, so the first term in the RHS of \eqref{eq:ctp-pair-balance} vanishes, and the boundary-match count is nonnegative.
Taking $t_0=0$ in~\eqref{eq:ctp-pair-balance} and using $E_k(0)=0$ therefore gives $E_k(t)\geq-O\!\left(R(U)\right)$.

Otherwise, $r_k\notin\calA_s^{\mathrm{ret}}$, so $\underline a_k=\underline a_{k-1}$ and compatibility gives $q_{r_k}=q_{u_k}\mathrel{=}0$. Hence $\bar q_{r_k}\mathrel{\equiv}0$; moreover, $u_k$ is empty at time zero.
We use the queue-length identity \eqref{eq:ctp-core-balance} at $r_{k+1}$ to evaluate the boundary-match count appearing in~\eqref{eq:ctp-pair-balance}.
Let $t_0$ be the last empty time of $u_k$ in $[0,t]$.
If $t_0=t$, then $Q_{u_k}(t)=0$ and $\bar q_{r_k}(t)=0$, so~\eqref{eq:ctp-pair-inventory} gives
\[
 E_k(t)-E_{k-1}(t)=Q_{r_k}(t)\geq0.
\]
Hence $E_k(t)\geq E_{k-1}(t)\geq-W_{k-1}$ in this case.

For the remainder of the lower-bound argument, suppose $t_0<t$.
The retained arrival at step $t_0+1$ finds both cycle neighbors empty, so $Q_{r_k}(t_0)=Q_{r_{k+1}}(t_0)=0$ and
\begin{equation}\label{eq:ctp-positive-u}
 E_k(t_0)=E_{k-1}(t_0).
\end{equation}
The queue at $r_{k+1}$ stays empty throughout $[t_0,t]$, and arrivals there choose an attached child or the occupied $u_k$ before its other cycle neighbor, so no match uses the edge $\{r_{k+1},r_{k+1}-1\}$.
Since $Q_{r_{k+1}}(t_0)=Q_{r_{k+1}}(t)=0$ and $z_{r_{k+1}}=a_k$, subtracting equation~\eqref{eq:ctp-core-balance} at time $t_0$ from the same equation at time $t$, with $i=r_{k+1}$ gives
\[
 \begin{aligned}
 &\#\{\text{matches across }\{u_k,r_{k+1}\}\text{ in }(t_0,t]\}\\
 &\quad=(z_{r_{k+1}-1}+a_k)(t-t_0)
 +\sum_{j\in\calT_{r_{k+1}}}(-1)^{d(r_{k+1},j)}\bigl(M_j(t)-M_j(t_0)\bigr)
 +Y_{r_{k+1}}(t)-Y_{r_{k+1}}(t_0).
 \end{aligned}
\]
The absolute value of the sum involving $M_j(t)-M_j(t_0)$ is at most $\sum_{j\in\calT_{r_{k+1}}}R_j(U)\leq R(U)$.
The remaining difference has absolute value at most $2\sup_{0\leq t\leq U}|Y_{r_{k+1}}(t)|=O(R(U))$, by~\eqref{eq:ctp-branch-imbalance}.
Substituting this boundary-match count into~\eqref{eq:ctp-pair-balance} therefore gives
\begin{equation*}
 \begin{aligned}
 E_k(t)-E_k(t_0)
 &\mathrel{=}\bigl[-(a_k-\underline a_k)+(z_{r_{k+1}-1}+a_k)\bigr](t-t_0)+O\!\left(R(U)\right)\\
 &=(z_{r_{k+1}-1}+\underline a_k)(t-t_0)
   +O\!\left(R(U)\right).
 \end{aligned}
\end{equation*}
The coefficient $z_{r_{k+1}-1}+\underline a_k$ is positive, so~\eqref{eq:ctp-positive-u} gives
\begin{align*}
 E_k(t)\geq E_{k-1}(t_0)-O\!\left(R(U)\right)
 \geq-W_{k-1}-O\!\left(R(U)\right),
\end{align*}
as desired.
\end{proof}

\subsection{Proof of Lemma~\ref{lem:ctp-refresh}}\label{sec:proof-refresh}

Suppose $Q(0)=q$. Since $\eta=2\epsilon\Delta$, Proposition~\ref{lem:ctp-core} gives, for a fixed universal constant $C>0$,
\begin{equation}\label{eq:ctp-fixed-block-comparison}
 \norm{Q(\Delta)}_1\leq\norm{q}_1-\min\{\mu(q),\eta\}+Cn^4R(\Delta).
\end{equation}
By~\eqref{eq:prelim-range-first-moment} and~\eqref{eq:prelim-range-second-moment}, there is a universal constant $C'$ such that, for $K\geq\max\{4C'^2,4C'\}$,
\begin{align}
 Cn^4\Expect R(\Delta)&\leq C'n^{9/2}\sqrt\Delta\leq \eta/4,\label{eq:ctp-error-first-moment}\\
 C^2n^8\Expect[R(\Delta)^2]&\leq C'n^9\Delta\leq \eta^2/16.\label{eq:ctp-error-second-moment}
\end{align}
The final inequalities use $\eta/4=\epsilon\Delta/2$ and $\Delta\geq Kn^9/\epsilon^2$.
The elementary bound
\begin{equation}\label{eq:ctp-core-cost}
 \norm{q}_1\min\{\mu(q),\eta\}\geq \eta\norm{q}_1-\eta B(q)-n\eta^2
\end{equation}
holds because it is immediate when $\mu(q)\geq\eta$, while $\mu(q)<\eta$ implies $\lVert q\rVert_1-B(q)\leq n\mu(q)<n\eta$ and makes the right-hand side negative.
The right-hand side of~\eqref{eq:ctp-fixed-block-comparison} is nonnegative because $\min\{\mu(q),\eta\}\leq\mu(q)\leq\lVert q\rVert_1$.
Squaring~\eqref{eq:ctp-fixed-block-comparison}, using $\min\{\mu(q),\eta\}^2\leq\lVert q\rVert_1\min\{\mu(q),\eta\}$, and applying~\eqref{eq:ctp-error-first-moment}, \eqref{eq:ctp-error-second-moment}, and~\eqref{eq:ctp-core-cost} yields
\[
 \begin{aligned}
 \Expect_q\norm{Q(\Delta)}_1^2-\norm{q}_1^2
 &\leq-\norm{q}_1\min\{\mu(q),\eta\}
   +2Cn^4\norm{q}_1\Expect R(\Delta)
 +C^2n^8\Expect[R(\Delta)^2]\\
 &\leq-\frac{\eta}2\norm{q}_1+\eta B(q)+\left(n+\frac1{16}\right)\eta^2,
 \end{aligned}
\]
which proves~\eqref{eq:ctp-refresh-square}.
For~\eqref{eq:ctp-refresh-conditional}, the fixed priority rule and i.i.d.\ arrivals make $Q$ a time-homogeneous Markov chain.
Conditionally on $Q(t_0)$, the process $(Q(t_0+u))_{u\geq0}$ therefore has the law of the queue process initialized at $Q(t_0)$, so~\eqref{eq:ctp-refresh-square} applies with $q=Q(t_0)$.

\end{document}